\documentclass[11pt, reqno]{amsart}
\usepackage{amsmath,amssymb,amsfonts,amscd,hyperref,color}
\usepackage[utf8]{inputenc}
\usepackage{csquotes}
\usepackage{nicefrac}

\usepackage[abs]{overpic}
\usepackage{verbatim}
\usepackage{enumerate}
\usepackage{tikz}

\newcommand{\Hmm}[1]{\leavevmode{\marginpar{\tiny%
$\hbox to 0mm{\hspace*{-0.5mm}$\leftarrow$\
	hss}%
\vcenter{\vrule depth 0.1mm height 0.1mm width \the\marginparwidth}%
\hbox to
0mm{\hss$\rightarrow$\hspace*{-0.5mm}}$\\\relax\raggedright #1}}}

\newtheorem{theorem}{Theorem}
\newtheorem{corollary}[theorem]{Corollary}
\newtheorem{lemma}[theorem]{Lemma}
\newtheorem{proposition}[theorem]{Proposition}

\theoremstyle{definition}

\newtheorem{example}[theorem]{Example}
\newtheorem*{remark}{Remark}

\numberwithin{equation}{section}
\newcommand{\Z}{{\mathbb Z}}
\newcommand{\R}{{\mathbb R}}

\newcommand{\N}{{\mathbb N}}

\newcommand{\D}{{\mathbb D}}

\let\L\undefined

\newcommand{\L}{\mathcal{L}}
\newcommand{\G}{\mathcal{G}}
\newcommand{\F}{\mathcal{F}}
\renewcommand{\P}{\mathcal{P}}
\newcommand{\al}{{\alpha}}
\newcommand{\be}{{\beta}}

\newcommand{\ph}{{\varphi}}
\newcommand{\eps}{{\varepsilon}}

\newcommand{\abs}[1]{\left| #1 \right|}
\newcommand{\norm}[1]{\left\| #1 \right\|}

\renewcommand{\d}{\,\mathrm{d}}
\newcommand{\ep}{\varepsilon}

\newcommand{\s}{\sigma}
\renewcommand{\L}{\mathcal{L}}
\renewcommand{\D}{\mathcal{D}}

\newcommand{\set}[1]{\left\lbrace #1 \right\rbrace }

\begin{document}
\title[Positive Criticality and Optimal Hardy Inequality]
{Positive Criticality and optimal Hardy Inequality for Fractional Laplacians}
\author[P.~Hake]{Philipp Hake}

\address{P.~Hake, Institut f\"ur Mathematik, Universit\"at Leipzig, 04109 Leipzig, Germany}\email{philipp.hake@math.uni-leipzig.de}
\author[M.~Keller]{Matthias Keller}
\address{M.~Keller, Israel Institute of Advanced Studies, Jerusalem, Israel; Institut f\"ur Mathematik, Universit\"at Potsdam
14476  Potsdam, Germany}
\email{matthias.keller@uni-potsdam.de}
\author[F.~Pogorzelski]{Felix Pogorzelski}
\address{ F.~Pogorzelski, Israel Institute of Advanced Studies, Jerusalem, Israel;  Institut f\"ur Mathematik, Universit\"at Leipzig, 04109 Leipzig, Germany}
\email{felix.pogorzelski@math.uni-leipzig.de}

\date{\today}

\begin{abstract} We characterize positive critical Hardy weights for general Laplacians on weighted graphs. We then apply this result to fractional Laplacians on general graphs and use the characterization to identify an optimal Hardy weight under suitable assumptions. We finally illustrate our results with examples of graphs which arise as Cayley graphs of groups, satisfy curvature assumptions or are fractal graphs.
\end{abstract}
\maketitle

\tableofcontents

\section{Introduction and main results}
The study of Hardy inequalities is a classical topic over the last century \cite{Davies,HardyHistory} which has seen tremendous attention for discrete and fractional operators especially in the last decades see e.g. \cite{BD,BSV,CGL,CR18,DFF,DKP26,Dyda,Fischer,FKP23,FLS,FSW08,FR25,GKS25,Go,Gup23,Gup24,KL16,KL23,KN,KPP18,KPPHardy,KPP20HardyOnManifoldsAndGraphs,KrSt,L20,L23,RS09,RS13,RS}. Besides their intrinsic interest, Hardy inequalities are a powerful tool to study spectral properties of operators and have applications in the study of partial differential equations, geometry and mathematical physics.

Throughout the last century these inequalities have been mainly studied as a fundamental tool in the analysis of partial differential equations. However, the recent interest in their original discrete form was made possible by the fundamental new perspective developed by \cite{DFP14,DP16}. These works answer a question raised by Agmon for optimal Hardy inequalities for general second order elliptic operators on domains in $\R^d$ and Riemannian manifolds, see \cite{Agmon}. Their supersolution construction allows to construct optimal Hardy weights from positive superharmonic functions. This method was transfered to graphs in \cite{KPP18} and allowed to construct optimal Hardy weights for the standard lattice graph over $\Z^d$. However, it has to be noted that while these weights have the anticipated asymptotics near infinity including the correct constant, they have higher order terms which are not positive except for the one-dimensional case \cite{KPPHardy,FKP23,KrSt,L20,RS}. Indeed, the comparison with the asymptotics of the optimal constant for the Hardy weight without higher order terms shows that these terms are necessary for the inequality to hold, cf.~\cite{Gup23}.

A second source of interest stems from the study of the fractional Laplacian which is a fundamental operator in geometric analysis, relativistic quantum mechanics, turbulence, elasticity, laser physics, and anomalous transport. While the optimal constant for the fractional Hardy inequality in Euclidean space is known for already a while \cite{BD,FLS,H77,Y99}, it was so far not studied in the continuum in regards for optimality of the Hardy weight. With the recent interest in its discrete counterpart \cite{CRSTV,DMRM,FRR,ZLY} the question of optimality arose as well \cite{CR18,Dyda} and was answered in the Euclidean setting \cite{DFF,HKP,KN} using criticality theory. These approaches use explicit knowledge of the asymptotics of the Riesz kernel  \cite{CR18,CRSTV,DMRM,FRR,GRM,Slade}.

In this paper we study the problem of critical Hardy inequalities on general graphs and treat then the case of the fractional Laplacian. Optimality of  a Hardy weight in the sense of \cite{DFP14,KPP18} assumes that the Hardy weight is maximal, i.e., is critical, and that the inequality does not have a minimizer, i.e., is null-critical.  Hardy weights which are critical but where the inequality allows for a minimizer are called positive critical. 

In the first part of the paper we characterize these positive critical Hardy weights for general graphs, see Theorem~\ref{thm:classification_positive-critical_weights_Hardy} below.  This result was already used as a specific tool in \cite{HKP}, where the focus is on determining the precise constant  in the specific setting of $ \Z^d $ rather than developing the general theory. Here, we are interested in the framework of general graphs and, therefore, present the result in full generality.

While positive critical Hardy weights are not optimal, the characterization  can be used to identify optimal Hardy weights. This is demonstrated for the case of fractional Laplacians later. For these operators, we study a family of Hardy weights arising from the Riesz kernel and characterize when they are positive critical, see Theorem~\ref{t:poscritical} for general fractional graph Laplacians. To identify null-critical and thus, optimal Hardy weights we need additional assumptions on the graphs.
Next, to a mild lower exponential bound on the measure we need Ahlfors regularity and two-sided Gaussian bounds with respect to one reference vertex which  is used to determine the asymptotics of the Riesz kernel. This control allows for an exact identification of the open interval of parameters for which the weight is positive critical. We then use the minimizers  to construct a null-sequence at the upper end of the interval which shows criticality of the corresponding weight. However, as we already know that it cannot be posivite critical, it must be null-critical. These results are given in Theorem~\ref{thm:MAIN2} below.

Let us introduce the setting to state the results explicitly.  Let $m:X\to (0,\infty)$ be a positive function on a countable set $X$. We call the pair $(X,m)$ a \emph{discrete measure space} where $m$ induces a measure of full support, defined for a set $A\subseteq X$ as $m(A)=\sum_{x\in A}m(x)$. 
%We let $\ell^2(X,m)$ be the associated real Hilbert space of square summable functions with respect to $m$. 

A \emph{graph} $(b,c)$ over a discrete measure space set $(X,m)$ is a symmetric function $b:X\times X\to [0,\infty)$ with zero diagonal satisfying $$\deg(x)=\sum_{y\in X}b(x,y)<\infty$$ for all $x\in X$ and a {\em killing term} $c:X\to[0,\infty)$. Throughout this work we will assume that graphs are {\em connected}, i.e.,\@ for any two $x,y \in X$ there are $x_1=x, x_2, \dots, x_n, x_{n+1}=y$ such that $b(x_i, x_{i+1})>0$ for all $1 \leq i \leq n$. Denoting by $C(X)$
the vector space of all real-valued functions on $X$,
the graph gives rise to the formal Laplacian $\L:\mathcal{F} \to C(X)$ acting on the space of functions $$\mathcal{F}=\{f:X\to \R\mid \sum_{y\in X} b(x,y)|f(y)|<\infty\mbox{ for all }x\in X\},$$ via 
\[\L f(x)=\frac{1}{m(x)}\sum_{y\in X}b(x,y)(f(x)-f(y))+\frac{c(x)}{m(x)}f(x)\] 
and the associated quadratic form $\mathcal{Q}$ with domain $\mathcal{D}$ given by
\[\mathcal{Q}(f)=\frac{1}{2}\sum_{x,y\in X}b(x,y)(f(x)-f(y))^2+\sum_{x\in X}c(x)f(x)^2,\] on the {\em functions of finite energy} \[\mathcal{D}=\{f:X\to \R \mid \mathcal{Q}(f)<\infty\}.\]
We write $C_c(X) \subseteq C(X)$ for the subspace of all finitely supported functions and let $\mathcal{D}_0$ be the {\em extended space} defined as the closure of $C_c(X)$ in $\mathcal{D}$ with respect to the norm $\|f\|_{o}=\sqrt{\mathcal{Q}(f)+|f(o)|^2}$ for some fixed $o\in X$, i.e., 
\begin{align*}
\mathcal{D}_0=\overline{ C_c(X)}^{\|\cdot\|_{o}}.
\end{align*}
%A function $u $ is called \emph{superharmonic} if $\L u\ge 0$ pointwise. 

Given a summable or positive function $ f $ over a set $ A $ (typically a subset of $X$ or $ X\times X $), we write $$  \sum_{A}f=\sum_{x\in A}f(x).$$

A  non-negative function $w:X \to [0,\infty)$ is called a \emph{Hardy weight} if
\[\mathcal{Q}(\varphi)\ge \sum_{X}wm\varphi^2\]
for all $\varphi\in C_c(X)$. Whenever the inequality holds, it can be extended to all $\varphi\in \mathcal{D}_0$.
 Furthermore, we call the graph and the corresponding Laplacian \emph{transient} if  there exists a non-trivial Hardy weight.	As is well-known \cite[Theorem~6.1]{KLW} or \cite[Theorem 5.3]{KPP20}, a graph $(b,c)$ over $(X,m)$ satisfies a Hardy inequality if and only if for every $x \in X$,  there is a unique function $G_x \in \D_0$ such that $\mathcal{L} G_x = 1_x$, where $1_x$ denotes the characteristic function of the set $\{x\} \subseteq X$. The function $G:X \times X \to (0,\infty)$, $ G(x,y) = G_y(x)$ is called the {\em Green's function} of $b$ associated with $\mathcal{L}$. 
 % It satisfies the symmetry relation $G(x,y)m(x) = G(y,x)m(y)$ for all $x,y \in X$, cf.\@ \cite[Theorem~6.26]{KLW}. 
The Green's function gives rise to a positivity preserving  \emph{Green's operator} $G$  on
$\mathcal{G}$ associated with $\L$, i.e.,
\begin{align*}
	Gk(x)&=\sum_{y\in X}G(x,y)k(y), \\
	\mathcal{G} &=\{k:X\to \R \mid\sum_{y\in X}G(x,y)|k(y)|< \infty\}.
\end{align*}

A Hardy  weight $w$ is {\em critical} if there is no Hardy  weight $w^{\prime} \geq w$ with $w^{\prime} \neq w$. By criticality theory \cite{KPP20}, criticality of $w$ is equivalent to the existence of a {\em ground state} $v \in \mathcal{F}$, i.e.,\@ a up to multiplication by scalars unique strictly positive function $v \in \mathcal{F}$ satisfying $\mathcal{L}v = wv$. A critical Hardy  weight is said to be {\em positive critical} if the ground state $v$ is square integrable with respect to the measure $wm$, which means that above inequality allows for a minimizer. A critical Hardy weight is {\em null-critical} if it is not positive critical.

The first main result is the following characterization of positive critical Hardy  weights. We denote by $\ell^1(X,m)$ the space of all absolutely summable functions in $C(X)$ with respect to the measure $m$.

\begin{theorem}[Positive critical Hardy  weights] \label{thm:classification_positive-critical_weights_Hardy}
	Let  $w:X \to [0, \infty)$ be a  function. Then the following are equivalent. 
	\begin{enumerate}[(i)]
		\item $w$ is a positive-critical non-trivial Hardy weight. 
		\item $w = \mathcal{L} v / v$ for some strictly positive, superharmonic and non-harmonic $v \in \D_0$ with $\mathcal{L} v^2 \in \ell^1(X,m)$.
		\item $w =\mathcal{L} v / v$ for some strictly positive, superharmonic and non-harmonic $v \in \D_0$.
		\item $w= {\mathcal{L} Gk}/{Gk} = {k}/{Gk}$ for some non-trivial, non-negative $k \in \mathcal{G}$  such that 
		\[
		\sum_{X} (G k) k \, m  < \infty.
		\]
	\end{enumerate}
\end{theorem}

We apply the theorem above to the fractional Laplacian. To this end, we consider the Dirichlet Laplacian $L$ for a graph $(b,c)$ on $\ell^2(X,m)$ which is the self-adjoint operator associated with the closure of the quadratic form $\mathcal{Q}$ defined on $C_c(X)$ in $\ell^2(X,m)$, see \cite[Chapter~1]{KLW}. 
The \emph{fractional Laplacian} $ L^{\sigma} $, $ \sigma\in(0,1] $ is then defined via the spectral theorem and, in particular, $ L^{1}=L  $. 
We will show in Section~\ref{sec:fractional} below that $L^{\sigma}$ is the restriction of a formal Laplacian $\mathcal{L}^{\sigma}$ associated with a connected graph $(b_{\sigma},c_{\sigma})$ over $(X,m)$,
where $b_{\sigma}$  and $c_{\sigma}$ have explicit representations via the  time integral of the semigroup $ e^{-tL} $ of $ L $.  
We refer to Theorem~\ref{t:fractionalgraph} below for the details.

To construct Hardy weights for the fractional Laplacian, we will use the following family of functions which are intimately related to the Green's functions of fractional Laplacians.
For fixed $ o\in X $, we consider the following family functions 
 $k_\alpha:X\to [0,\infty] $,
$\alpha>0$, by
\[
k_{\alpha} = k_{\alpha,o} = \frac{1}{|\Gamma(-\al)|}\int_{0}^{\infty}e^{-tL}1_o\frac{\d t}{t^{1+\al}}
\] 
We will observe in Lemma~\ref{lem:amax} below that $k_{\alpha}$ maps to $(0, \infty)$ for all $\alpha \in (0,\alpha_{\mathrm{max}})$, where 
\[
\alpha_{\mathrm{max}} = \sup\big\{ \alpha > 0 \mid k_{\alpha,o}(o) < \infty \big\}.
\] 

Applying  Theorem~\ref{thm:classification_positive-critical_weights_Hardy}, we obtain the following characterization of positive criticality of the aforementioned Hardy weights. 

\begin{theorem}[Positive critical Hardy weights of the fractional Laplacian]\label{t:poscritical}
	Let $(b,c)$ be a connected graph over $ (X,m) $, $ \sigma\in (0,1) $ and $\alpha\in (\sigma,\alpha_{\max})$. Then, $$w_{\sigma,\al} =\frac{k_{\al-\sigma}}{k_{\al}}$$ is a positive critical Hardy weight for $L^\sigma$ if and only if	
	\begin{align*}
		\sum_X {k_{\al-\sigma}}{k_{\al}} \, m<\infty.
	\end{align*}    
\end{theorem}

We will use the latter result to show that under the assumption of suitable geometric conditions and heat kernel bounds, there is a critical parameter $\alpha_0 \in (\sigma,\alpha_{\max})$ such that $w_{\sigma,\alpha_0}$ is a null-critical Hardy weight. We briefly explain these assumptions and refer to Section~\ref{sec:asymp} for details. 

We suppose  that the graph $(b,c)$ over $(X,m)$ is endowed with an intrinsic metric $\varrho$ with jump size $1$, see \cite[Chapter~11]{KLW}. We denote the corresponing distance balls with respect to a fixed vertex $o \in X$ by $ B_r(o) =\{x\in X\mid |x|\leq r\}$, $ r\ge 0 $, where $|x|=\varrho(o,x)$. We say that the measured metric space $(X,m,\varrho)$
 is {\em  Ahlfohrs regular of dimension $d$} at $ o $ if for $r > 0$ large enough, we have 
\begin{align}\label{e:AR}\tag{A}
	m\big( B_r(o)\big) \, \asymp \, r^{d}
\end{align}
where the symbol $\asymp$ means two-sided estimates by positive constants.  Furthermore, we need to assume that the measure $ m $ is \emph{sub-exponentially bounded from below}, i.e., $\liminf_{|x|\to\infty}\frac{1}{|x|} \log m(x)\geq 0  $. Finally, we assume two-sided long time bounds for the heat kernel  with respect to $\varrho$, 
%an intrinsic metric $\varrho$ with jump size $1$, 
and parameters $d\ge 1,\beta\ge 2$,  i.e.\@ there are constants $c_1,C_1,c_2,C_2>0$ such that for all $x \in X$ with $|x| \geq 1$ and $t \geq |x|$,
\begin{align}\label{e:GB}\tag{HB}
\frac{C_1}{t^{\frac{d}{\beta}}}  \exp\left(-c_1\left(\frac{|x|^{\beta}}{t}\right)^{\frac{1}{\beta-1}} \right) \leq  p_t(x,o) \leq \frac{C_2}{t^{\frac{d}{\beta}}}    \exp\left(-c_2\left(\frac{|x|^\beta}{t}\right)^{\frac{1}{\beta-1}} \right).
\end{align}
For given $ \sigma\in (0,d/\beta) $, we define 
\[
\alpha_0 = \frac{d + \sigma\beta}{2\beta}. 
\]

Our next main result is the following. 

\begin{theorem} \label{thm:MAIN2}
	Assume that $(b,c)$ is a connected graph over $(X,m)$ with measure $ m $ sub-exponentially bounded from below and let $o \in X$, $d \ge 1$, $ \beta\ge 2 $, $\sigma \in (0,1]$. Assume \eqref{e:GB}, then $ \alpha_{\max}=d/\beta $ and the fractional Laplacian  is transient if and only if $ \sigma<d/\beta $. Moreover,  for every $\alpha \in (\sigma,d/\beta)$, the function
	\[
	w_{\sigma,\alpha} = \frac{k_{\alpha-\sigma}}{k_{\alpha}}
	\] 
	is a strictly positive Hardy weight with asymptotics $w_{\sigma,\alpha} \asymp |x|^{-\beta\sigma}$. If the graph is Ahlfors regular \eqref{e:AR}, then
	\begin{itemize}
		\item[(a)] $w_{\sigma,\alpha}$ is positive critical for $\sigma < \alpha < \alpha_0$, 
		\item [(b)] $w_{\sigma, \alpha}$ is null-critical for $\alpha= \alpha_0$. 
	\end{itemize} 
\end{theorem}
 
\begin{remark}
	Observe that in the case when the graph itself is transient, the theorem above gives a strictly positive optimal Hardy weight for the classical Laplacian, i.e., $ \sigma=1 $.
\end{remark}

The proofs of these results are based on Theorem~\ref{thm:classification_positive-critical_weights} together with a careful analysis of  bounds of the Riesz kernels $k_{\alpha}$ under the respective conditions on the heat kernel, see Theorem~\ref{theorem:RieszAsymptotics}. 
\bigskip

{\bf Organization of the article.}  
In Section~\ref{sec:potential} we study connections between the extended space, the space of potentials and the space of functions vanishing at infinity under suitable integrability conditions on the Laplacian of a function. These results are applied in Section~\ref{sec:positivecriticality} for the proof of the characterization of positive criticality stated in Theorem~\ref{thm:classification_positive-critical_weights_Hardy}. 
These considerations are essentially found in the dissertation thesis of the first named author, see \cite[Chapter~1 and 2]{hake2025optimal}, and are  extended here to the situation of arbitrary measures $m$ and Laplacians with a killing term $c$. 
Section~\ref{sec:fractional} is devoted to the discussion of the fractional Laplacian over a general graph. Using  the first main result, we prove Theorem~\ref{t:poscritical} there. Bounds on the Riesz kernels under the assumption of  Ahlfors regularity and suitable heat kernel bounds are carried out subsequently. We thereby prove Theorem~\ref{thm:MAIN2}. Finally, we discuss classes of graphs for which the assumptions of Theorem~\ref{thm:MAIN2} are satisfied.

	\section{Extended space,  potentials and functions vanishing at infinity} \label{sec:potential}
	
	In this section, we relate the space of   potentials, which is the range of the Green operator,  the extended space $\D_0$ and the functions vanishing at infinity $ C_0(X) $. We will throughout assume that all considered graphs are transient. To set the stage, we recall some well known facts for transient graphs and which we expand upon below.
	
	The first important observation is that for a transient graph, the Green's functions $G(\cdot,o)=G1_o$ are superharmonic potentials which are in the extended space $\D_0$. Beyond this there are two important decompositions. First there is the Royden decomposition which states that every function of finite energy can be decomposed into a part in the extended space and a harmonic part. Secondly, the Riesz decomposition of a positive superharmonic function is a decomposition into a potential and a harmonic part. This foreshadows the connection between the extended space  and potentials in regard to superharmonic functions. Furthermore, $ \mathcal{D}_0 $ and $ C_{0}(X) $ are closures of $ C_c(X) $ with respect to different norms, and the question arises how these spaces are related. We will see that functions in $ C_{0}(X) $ with integrable Laplacian are automatic potentials in the extended space.

\subsection{Preliminaries} \label{sec:prelim}
In this subsection, we give further definitions of central objects and clarify some terminology relevant for this paper. We also collect some known assertions that will be used repeatedly in the subsequent sections below. 

\medskip

Let $ X $ be a discrete countable set.
We write $C_c(X)$ and $C_0(X)$ for the subspaces of $ C(X) $ of functions with finite support and of functions vanishing at infinity, respectively. For $f,g \in C(X)$, we write $f \geq g$,  $ f>g $ (or equivalently $ g\le f $, $ g<f $) if the corresponding inequalities hold pointwise and $  f\gneq g $ or $ g\lneq f $, if $ f\ge g $ and $  f\neq g $. If $ f\ge0 $ (respectively $ f>0 $), then $ f $ is called \emph{positive} (respectively \emph{strictly positive}).
We denote by $f_\pm = \max(0,\pm f)$ the {\em positive and negative part} of  $f \in C(X)$. Given a subset $A \subseteq X$ we denote by $1_A \in C(X)$ the function taking the value $1$ when $x \in A$, and $0$ otherwise  and  $1_x = 1_{\{x\}} $ for $ x\in X $.

Given a function $m:X \to (0,\infty)$ and $1 \leq p  < \infty$, we define
\[
\ell^p(X,m) = \big\{ f \in C(X)\mid \|f\|_p < \infty \big\}
\]
with norm
\[
\| f \|_p = \|f\|_{\ell^p(X,m)} = \Big( \sum_{ X}m|f|^p  \Big)^{1/p}. 
\]
%For $m$ being a measure of full support, one obtains a normed space.

Given a transient connected graph $(b,c)$ over $(X,m)$ with formal Laplacian $ \mathcal{L} $ on $\mathcal{F}$, we call $ \mathcal{L}u $ the \emph{charge} of  a function $u \in \mathcal{F}$. Further, $u$ is called {\em superharmonic} if $\mathcal{L}u \geq 0$, {\em subharmonic} if $-u$ is superharmonic and harmonic if $ u $ is  both super- and subharmonic.  
If the graph is transient, a function $u \in \mathcal{F}$ is a {\em potential} if  $ \mathcal{L}u\in \mathcal{G} $ such that
\[
u= G \mathcal{L} u, 
\] 
where $ G $ is the Green's operator on $ \mathcal{G}=\{u\in \mathcal{F}\mid G|u|<\infty\} $. We denote the set of  potentials by $\mathcal{P}$. 

\begin{proposition}
	\label{prop:characterisation_potentials} We have $$  \mathcal{P}=G\mathcal{G} . $$
	In particular, $ u $ is a potential if and only if there is $ k\in \mathcal{G} $ such that $ u=Gk $. In this case, $ \mathcal{L}u=k $ and $ u=G\mathcal{L}u $.
\end{proposition}

\begin{proof}
	The inclusion $ \mathcal{P}\subseteq G\mathcal{G} $ follows by definition.
	 For the converse, 
	let $k \in \mathcal{G}$ and set $u=Gk$. By splitting $k$ into positive and negative part, we may assume $k \geq 0$, and so $u\geq 0$.
	A short calculation using the identities $ u=Gk=\sum_xk(x)G1_x $ and $\mathcal{L}G1_x = 1_x$ shows 
	\[
	\sum_{y \in X} b(x,y) u(y) = \deg(x)u(x) + c(x)u(x) -m(x)k(x) .
	\]
	Thus, $u \in \mathcal{F}$ and $\mathcal{L} u = k$. With $k \in \mathcal{G}$, we obtain $u=Gk = G \mathcal{L}u$ and all involved sums  converge absolutely. 
	\end{proof}

	We finish this section with a characterization of the space $\D_0=\overline {C_{c}(X)}^{\|\cdot\|_{o}}$. 
	
	\begin{lemma}\label{prop:charD_0}
		We have $u \in \D_0$ if and only if there is a sequence $(\varphi_n)$ in $C_c(X)$ converging pointwise to $u$ as $n \to \infty$ and satisfying
		\[
		\sup_{n} \mathcal{Q}(\varphi_n) < \infty.
		\]
	\end{lemma}
	\begin{proof}In  \cite[Lemma~6.5]{KLW} it is shown that convergence in $ \mathcal{D}_0 $ is characterized by $ \ph_n\to u $ and $ \limsup_{n\to\infty}\mathcal{Q}(\ph_n) \leq \mathcal{Q}(u)$.
	Thus, the ``only if'' direction follows immediately. 
	The ``if'' direction follows by the Banach-Saks theorem  along the line of the proof of \cite[Lemma~6.12]{KLW} (iii) $ \Longrightarrow $ (i): Since $ \mathcal{D}_0 $ is a Hilbert space, \cite[Lemma~6.3~(d)]{KLW}, the bounded sequence $ (\ph_n) $ has a weakly convergent subsequence which has the limit $u $ due to the pointwise convergence. By the Banach-Saks theorem, this subsequence has a strongly convergent subsequence of Ces\`aro means with limit $u$ as well. Hence, we obtain that $u \in \mathcal{D}_0$.
	\end{proof}

	\subsection{Potentials and integrable charges} 
	
	The goal of this section is to study  some relations between the extended space $\D_0$ and the set of potentials $\mathcal{P}$ for a connected transient graph $ (b,c) $ over $ (X,m) $. The main result, Propostion~\ref{prop:equivalence_P_D_0} says that a function $u$ in $ \mathcal{P} $ or $ \D_{0} $ can be decomposed into the difference of two potentials whenever the (negative part of the)  charge $ \mathcal{L}u $ of $ u $ is in the space $ \mathcal{G}_2 $, which is defined next.
	
	\medskip
	
	We define
	\begin{align*} 	%\label{def:G_2}
		\mathcal{G}_2 = \{k \in \mathcal{G}\mid \sum_{ X}m|k| G|k| < \infty\}
	\end{align*}
	which is easily seen by the by elementary inequality $ (\al+\beta)^{2}\le2 \alpha^{2}+2\beta^{2} $, $ \al,\be\in\mathbb{R} $ to be a vector space. The next lemma gives some more structural insight.
	
	\begin{lemma} \label{lem:G_2}
		We have $ C_c(X) \subseteq \mathcal{G}_2$, 
		$$ G\mathcal{G}_2\subseteq \D_0 \cap \mathcal{P}$$ and for $ k\in \mathcal{G}_2$, we have $ k_{\pm}\in \mathcal{G}_{2} $ and
		\begin{align*}
			\mathcal{Q}(Gk) = \sum_{ X}m k Gk.
		\end{align*}
	\end{lemma}
	
	\begin{proof}
			We prove the formula for the energy first. Assume that $k \in  \mathcal{G}_2$ is such that $Gk\in \D_0$. Then by  \cite[Lemma~6.8]{KLW} 
			(which is applicable after splitting $ k $ into positive and negative part since $ Gk_{\pm}\in \mathcal{D}_{0} $ by \cite[Lemma~6.7]{KLW} and $ \mathcal{L}Gk_{\pm} =k_{\pm}\ge 0 $)
		\begin{align*}
			\mathcal{Q}(Gk) = \sum_{ X}m (\mathcal{L} Gk) Gk = \sum_{X}m k Gk,
		\end{align*}
where $ \mathcal{L}Gk=k $ follows by Proposition~\ref{prop:characterisation_potentials}. 
 Since $ G(\cdot,y) \in \mathcal{D}_0$ by \cite[Lemma 6.27]{KLW}, it follows $ Gk \in \mathcal{D}_0$ for every $ k\in C_c(X)$ and, therefore, the above formula gives $ C_c(X) \subseteq \mathcal{G}_2$.

Let $ k\in \mathcal{G}_2  $. We will  show that  $ Gk\in \D_0 $. To this end, we first show that $ k\in \mathcal{G}_2 $ implies $ Gk \in \mathcal{G}$. Note  that $\mathcal{G}_2$ is closed under taking the positive part or the negative part of a function, since 
			\begin{align*} %\label{eqn:posnegativepart}
				\abs{k(x)k(y)} %= (k_+(x)+k_-(x))(k_+(y)+k_-(y)) %\nonumber \\				&
				\geq k_+(x)k_+(y) + k_-(x)k_-(y).
			\end{align*}
		Hence, we can assume $ k\ge0 $.  If $k=0$, there is nothing to  show. So assume that there is some $o \in X$ with $k(o) > 0$.
		Let $(k_n)$ be a  sequence in $C_c(X)$ such that $0\le k_n\nearrow k$ pointwise montonously as $n \to \infty$. Then $(Gk_n)$ is also pointwise monotonically increasing. 		
		
		We claim that for every $z \in X$ the limit $(Gk_n(z))$ exists and is finite. Indeed, for $z \in X$, the local Harnack inequality \cite[Theorem~4.1]{KLW} yields a constant $C_{z,o} > 0$ such that $G(z,y) \leq C_{z,o} G(o,y)$ for every $y \in X$. Thus, for every $n \in \N$,
				\begin{multline*}
				m(o)C_{z,o}^{-1} k_n(o) Gk_n(z) = m(o) C_{z,o}^{-1} k_n(o) \sum_{y \in X} G(z,y) k_n(y) \\
				\leq \sum_{y \in X} m(o)  G(o,y)k_n(o) k_n(y) 
				\leq \sum_{x, y \in X}m(y) G(x,y) k(x) k(y) < \infty.
			\end{multline*}
			It follows that the pointwise limit of $(Gk_n)$ exists since $\lim_{n \to \infty} k_n(o) = k(o) > 0$.  With the monotone convergence theorem we can conclude  $ Gk_{n}\nearrow Gk $ pointwise monotonously. Since $G k_n\in \mathcal{D}_{0} $ as discussed above, the formula in the beginning of the proof yields $\sup_n \mathcal{Q}(Gk_n) \le \sum_{X}m k Gk<\infty $ as we assumed $k\in \mathcal{G}_2  $. Hence, by Lemma~\ref{prop:charD_0}, we obtain $ Gk \in \mathcal{D}_0$ and the formula formula in the beginning holds for $ k\in \mathcal{G}_2 $. Finally, $ Gk\in \mathcal{P} $ holds by definition and we have finished the proof.
	\end{proof}
	
		The characterization of positive criticality in Theorem~\ref{thm:classification_positive-critical_weights} relies on the following technical proposition.
	
	\begin{proposition}
		\label{prop:equivalence_P_D_0}
		Let $u \in \F$. Then the following assertions are equivalent.
		\begin{enumerate}[(i)]
			\item $u \in \mathcal{P}$ and $\mathcal{L} u \in \mathcal{G}_2$.
			\item $u \in \mathcal{P} \cap \D$ and $(\mathcal{L} u)_- \in \mathcal{G}_2$.
			\item $u \in \D_0$ and $(\mathcal{L} u)_- \in \mathcal{G}_2$.
			\item There are superharmonic functions $u^+,u^- \in \D_0$ such that $u = u^{+} - u^{-}$. Specifically, $ u^{\pm} $ can be chosen as the positive potentials $ G(\mathcal{L} u)_{\pm} $.
		\end{enumerate}
				In particular, if $u$ is superharmonic then $u \in \D_0$ if and only if $u \in \mathcal{P} \cap \D$. Moreover, if $u$ satisfies these equivalent conditions then it also holds that
		\begin{align*}
			\mathcal{Q}(u) = \sum_{ X}m(\mathcal{L} u)( G\mathcal{L} u).
		\end{align*}
	\end{proposition}
	
		To prove this proposition we will  make  use of the following lemma.
	
	\begin{lemma}
		\label{lem:superharmonic_and_D_0_is_potential}
		If $u \in \D_0$ is superharmonic, then $u \in \mathcal{P} \cap \D$,  $\mathcal{L} u \in \mathcal{G}_2$ and 
		\begin{align*}
			\mathcal{Q}(u) = \sum_{  X}m(\mathcal{L} u)( G\mathcal{L} u).
		\end{align*}
	\end{lemma}
	
		\begin{proof}
			Let $u \in \D_0$ be superharmonic. It is clear that $u \in \D$ since $\D_0 \subseteq \D$. 
		Combining $\mathcal{L} G_x =  1_x $, $ x\in X $, with 	 
		Green's formula, \cite[Lemma~6.8]{KLW}, for $G_x = G(\cdot,x) \in \D_0$, \cite[Lemma~6.27]{KLW} and the superharmonic $ u\in\mathcal{D} $ applied twice and the 	symmetry of the Green's function, \cite[Theorem~6.26]{KLW},    yields
		\begin{align*}
		u(x)m(x)=\sum_{ X}m \mathcal{L} G_x u  =	\mathcal{Q}(G_x, u)=\sum_{X}mG_x\mathcal{L}u = m(x) G\mathcal{L}	u(x).
		\end{align*}
		Hence, $ u=G\mathcal{L}u $ is a potential. Another application of Green's formula \cite[Lemma~6.8]{KLW} applied for the superharmonic $ u\in \mathcal{D}_{0} $ yields
		\begin{align*}
			\mathcal{Q}(u) &= \sum_{ X} m (\mathcal{L} u)u = \sum_{ X}m  (\mathcal{L} u) (G\mathcal{L}u).
		\end{align*}
		Since $u$ has finite energy this implies that $\mathcal{L} u \in \mathcal{G}_2$ and the lemma is proven. 
	\end{proof}
	
	For the proof of Proposition~\ref{prop:equivalence_P_D_0}, we will also use the Royden decomposition from \cite{KLSW17}.

	\begin{proposition}[Royden decomposition, Proposition 5.1 in \cite{KLSW17}]
		\label{thm:Royden_decomposition}
		For any $u \in \D$ there exists a unique decomposition $u = u_0 + u_h$ with $u_0 \in \D_0$, $u_h \in \D$ and $u_h$ harmonic. If $u \geq 0$ then $u_h \geq 0$.
	\end{proposition}
	
	%\begin{proof}		See Proposition 5.1 in \cite{KLSW17} in combination with the observation that an arbitrary measure $m$ of full support does note change the transience properties involved in the statement.	\end{proof}

	We are now ready to give the proof of the Proposition~\ref{prop:equivalence_P_D_0}.
	
	\begin{proof}[Proof of Proposition~\ref{prop:equivalence_P_D_0}]
		Observe first that if we have established that $(ii) \Longleftrightarrow (iii)$, then we can conclude that a superharmonic $u$ is in $\D_0$ if and only if it is in $\mathcal{P} \cap \D$ since it satisfies $(\mathcal{L} u)_- = 0 \in \mathcal{G}_2$. 
		
		Next, if $u \in \F$ satisfies $(i)$, i.e., $u \in \mathcal{P}$ and $\mathcal{L} u \in \mathcal{G}_2$, then 
		\begin{align*}
			\mathcal{Q}(u) = \sum_{  X}m( \mathcal{L} u)(G\mathcal{L} u)
		\end{align*}
		by Lemma~\ref{lem:G_2} and Proposition~\ref{prop:characterisation_potentials}.
		It remains to prove the equivalences.

		We first prove $(i) \Longrightarrow (iv)$. If $\mathcal{L} u \in \mathcal{G}_2$ then $(\mathcal{L} u)_\pm \in \mathcal{G}_2$ by Lemma~\ref{lem:G_2}
		and we can define $u^\pm = G(\mathcal{L} u)_\pm$ which are positive since the Green's operator is positivity preserving, \cite[Theorem~6.26~(b)]{KLW}. Those are superharmonic functions which are in $\D_0$ by Lemma~\ref{lem:G_2}
		and Proposition~\ref{prop:characterisation_potentials}. 
		Since $u$ is a potential, we have		
		\begin{align*}
			u = G \mathcal{L} u = G \big( (\mathcal{L} u)_+ - (\mathcal{L} u)_- \big) = u^+ - u^-.
		\end{align*}
		Next, we prove $(iv) \Longrightarrow (iii)$. So, let $u = u^+ - u^-$ be a decomposition of $u$ into superharmonic functions in $\D_0$. Then $u \in \D_0$ and
		\begin{align*}
			0 \leq (\mathcal{L} u)_- = (-\mathcal{L} u)_+=  (\mathcal{L} u^- - \mathcal{L} u^+)_{+} \leq \mathcal{L} u^-
		\end{align*}
		since $\mathcal{L} u^+, \mathcal{L} u^- \geq 0$. Since,  $\mathcal{L} u^- \in \mathcal{G}_2$ by Lemma~\ref{lem:superharmonic_and_D_0_is_potential}, we conclude  $(\mathcal{L} u)_- \in \mathcal{G}_2$.
		
		To prove  $(iii) \Longrightarrow (ii) $ and $(ii) \Longrightarrow (i)$, we   set 
		$$ u^{-}=G(\mathcal{L} u)_- \quad\mbox{ and }\quad u^+ = u + u^- . $$ 

		\emph{Claim.} If $ (\mathcal{L} u)_- \in \mathcal{G}_2 $, then  $ u^{-}\in \mathcal{D}_0 \cap \mathcal{P}$ and   $$ \mathcal{L}u^\pm=(\mathcal{L}u)_{\pm}\ge 0.$$% i.e., $ u^{\pm} $ are superharmonic.

		\emph{Proof of the claim.} The first statement follows directly from by Lemma~\ref{lem:G_2}. Furthermore, we clearly have $\mathcal{L}u^-=(\mathcal{L}u)_{-}\ge 0$ by Proposition~\ref{prop:characterisation_potentials}. Moreover, 
		\begin{align*}
			\mathcal{L} u^+ = \mathcal{L} u + \mathcal{L} u^- = \mathcal{L} u + (\mathcal{L} u)_-  = (\mathcal{L} u)_+ \geq 0.
		\end{align*}
		This finishes the proof of the claim.\hfill\qedhere

		For $(iii) \Longrightarrow (ii)$, 	
		assume $ u\in \mathcal{D}_{0}  $ and $(\mathcal{L} u)_- \in \mathcal{G}_2$.  By the claim, $ u^{-}\in \mathcal{D}_{0} $, so   $u^+ = u + u^- \in \D_0$. Consequently, $u^+ \in \P$ by Lemma~\ref{lem:superharmonic_and_D_0_is_potential}. This yields $u = u^+ - u^- \in \P \cap \D$. By the claim,   $ u^{-}\in  \mathcal{D}_{0}\cap \mathcal{P} $, and, thus, 
		$u^+=u+u^{-} \in \P \cap \D$. We next use the Royden decomposition, Proposition~\ref{thm:Royden_decomposition}, to decompose $u^+ = u^+_0 + u^+_h$ with $u^+_0 \in \D_0$ and with a harmonic function $u^+_h \in \D$. Then, $\mathcal{L} u^+_0 = \mathcal{L} u^+ \geq 0$ and by the claim we obtain $(\mathcal{L} u)_+ = \mathcal{L} u^+ = \mathcal{L} u^+_0 \in \mathcal{G}_2$ by Lemma~\ref{lem:superharmonic_and_D_0_is_potential}. 	
		Since $\mathcal{G}_2$ is a vector space, it follows that $\mathcal{L} u = (\mathcal{L} u)_+ - (\mathcal{L} u)_- \in \mathcal{G}_2$. %That $u$ is in $\P$ is trivial.
	\end{proof}
	
	\subsection{Superharmonic functions vanishing at infinity}
		In this subsection, we  show that superharmonic functions in  $C_0(X)$ are extended space are potentials and show a Green's formula for functions in $ C_{0}(X) $ with integrable charge. %For $c=0$ and $m=1$, the following proposition can be found in \cite[Corollary~1.39]{hake2025optimal}.  
	
	\begin{proposition}
		\label{cor:superharmonic_in_C_0_is_potential}
		If $u \in C_0(X)$ is superharmonic, then $u$ is a potential.
	\end{proposition}
	
	The proof relies on the  Riesz decomposition for superharmonic functions.

	\begin{proposition}[Riesz decomposition, Theorem 2.1 in \cite{FK21}]
		\label{thm:Riesz_decomposition}
		If $u$ is superharmonic and if there exists a subharmonic function $f$ with $f \leq u$ then there exists a unique decomposition $u = u_p + u_h$ with a potential $u_p$ and a harmonic function $u_h$. Further $u_p \geq 0$ and $u_h \geq f$.
	\end{proposition}
	We also need the following lemma.
	
	\begin{lemma}
		\label{lem:superharmonic_in_C_0_is_positive}
		The Laplacian $\mathcal{L}$ is injective on $C_0(X)$ and if $u \in C_0(X)$ is superharmonic, then $u \geq 0$.
	\end{lemma}
	
	\begin{proof}
		Let $u \in C_0(X)$ be superharmonic. Assume   that there exists $o \in X$ with $u(o) < 0$. The set $\set{x \in X \mid u(x) \leq u(o)}$ is non-empty and finite because $u \in C_0(X)$. Hence, $u$ must attain a global minimum. By  the minimum principle \cite[Theorem~1.7]{KLW}, $ u $ is constant. Since $ u\in C_{0}(X) $, we have $ u=0 $ which is a contradiction to $u(o) < 0$.
		
		If $\mathcal{L} u = 0$ for some $u \in C_0(X)$ then $u$ and $-u$ are superharmonic so we must have $u = 0$ which is the desired injectivity.
	\end{proof}
	
		\begin{proof}[Proof of Proposition~\ref{cor:superharmonic_in_C_0_is_potential}]
		Let $u \in C_0(X)$ be superharmonic. 
		The previous lemma shows that $u \geq 0$. By the Riesz decomposition we have $u = u_p + u_h$ with  $0 \leq u_h \leq u$. Therefore the harmonic part $u_h$ is in $C_0(X)$ and, thus, trivial by the lemma above. Therefore, $u = u_p$ is a potential. 
	\end{proof}
	
We next show a Green's formula for functions in $C_0(X)$ with integrable charge.
	
	\begin{proposition}
		\label{lem:C_0_u_Delta_u_in_ell^1_implies_D_0}
		Let $u \in C_0(X)$ with $\mathcal{L} u \in \ell^1(X,m)$. Then, $u \in \D_0$ and
		\begin{align*}
			\mathcal{Q}(u) = \sum_{  X}m u\mathcal{L} u.
		\end{align*}
		In particular, $\mathcal{Q}(u) \leq \|u\|_{\infty} \|\mathcal{L} u\|_{1}$.
	\end{proposition}
	
		%For $c=0$ and $m=1$, the above proposition can be also be found in \cite[Lemma~1.40 and Corollary~1.41]{hake2025optimal}. 
	
	\begin{proof}The result follows from  \cite[Lemma~6.8]{KLW}, if we can show that $u \in \D_0$. 		Consider the normal contraction $ T_{n}:\R\to\R $
\begin{align*}
T_{n}(t)= \begin{cases}
				0, &t \in [-\frac{1}{n},\frac{1}{n}], \\
				t - \frac{1}{n}, &t > \frac{1}{n}, \\
				t + \frac{1}{n}, &t < -\frac{1}{n}.
			\end{cases}
\end{align*}
for $ n\in \N $ and let $ 
u_n=T_n \circ u $. One checks easily that $ |u_n|\le |u| $ and by case distinction that
\begin{align*}
(u_{n}(x)-u_{n}(y))^{2} \le (u(x)-u(y)) (u_n(x)-u_n(y)) .
\end{align*}
for all $ x,y\in X $.
		Since $u$ is in $C_0(X)$, each $u_n$ is in $C_c(X)$.  Hence, for every $n \in \N$, we obtain
		\begin{align*}
			\mathcal{Q}(u_n) \leq \mathcal{Q}(u, u_n) = \sum_{x \in X}mu_n \mathcal{L} u \leq \norm{u}_{\infty}\norm{\mathcal{L} u}_{1},
		\end{align*}
		where we used the Green formula for $u \in \F$, $u_n \in C_c(X)$, \cite[Proposition 1.5]{KLW}. Moreover, $\sup_{n \in \N} \mathcal{Q}(u_n) \leq \mathcal{Q}(u)  < \infty$ and since $(u_n)$ converges pointwise to $u$ we infer from Lemma~\ref{prop:charD_0}  that $u \in \D_0$. The ``in particular'' statement follows from Fatou's lemma. 
	\end{proof}

\section{Hardy type weights and positive criticality} \label{sec:positivecriticality}

We now use the potential theoretic aspects developed above to characterize positive critical Hardy type weights, thus establishing Theorem~\ref{thm:classification_positive-critical_weights_Hardy}.
%The main results in this section are taken from Chapter~2 of the dissertation thesis of the first named author, see \cite[Section~2.1]{hake2025optimal} and extended to general discrete measure spaces $(X,m)$ and Laplacians with a non-vanishing killing term $c$. 

\subsection{Critical Hardy type weights}

Let $(b,c)$ be  a connected graph over $(X,m)$. We extend the notion of Hardy weights to functions that are not necessarily positive. Precisely, $w:X \to \R$  is called a {\em Hardy type weight} if 
\[\mathcal{Q}(\varphi)\ge w(\varphi) =  \sum_{ X}mw\varphi^{2} \quad \mbox{ for all } \varphi \in C_c(X).
\]
Whenever the inequality holds, we say that $w$ satisfies a {\em Hardy type inequality}. If $w \geq 0$, then $w$ is a {\em Hardy weight} and satisfies a {\em Hardy inequality}. We call a Hardy type weight {\em critical} if for every Hardy type weight $w^{\prime} \geq w$, we have $w=w^{\prime}$ or equivalently if $ 0 $ is a critical Hardy weight for the Schrödinger operator $ \mathcal{L}-w $. Using criticality theory for graphs with a measure, cf.\@  \cite[Section~6]{KLW}, it is straight forward to adapt the characterization of criticality from \cite[Theorem~5.3]{KPP20} to our setting. 

\begin{theorem}[Criticality of Hardy type weights]
	\label{lem:characterisaton_criticality}
Let $w$ be a Hardy-type weight. Then the following are equivalent.
\begin{enumerate}[(i)]
	\item $w$ is critical.
	\item There exists a unique (up to scalar multiplication) function $v \in \F$ such that $v \gneq 0$ and $\mathcal{L} v \geq wv$. 
	\item There exists a sequence $(e_n)$ in $C_c(X)$ such that $(\mathcal{Q}-w)(e_n) \to 0$ and such that $(e_n(o))$ converges to a non-zero constant for some $o \in X$.
	\item There exists a sequence $(e_n)$ in $C_c(X)$ such that $(\mathcal{Q}-w)(e_n) \to 0$ and such that $(e_n)$ converges pointwise to a function $v \in \F$ with $v \gneq 0$ and $\mathcal{L} v = wv$. 
\end{enumerate}
If these equivalent conditions hold then any $v \in \F$ with $\mathcal{L} v \geq wv$ satisfies $v > 0$ and $\mathcal{L} v = wv$ and the sequence $(e_n)$ from $(iv)$ can be chosen such that $0 \leq e_n \leq v$ for all $n \in \N$. 
	\end{theorem} 
	
Given a critical Hardy type weight $w$, the function $v$ obtained from the characterizing property~$ (iv) $  in the above theorem is unique up to scalar multiplication. We call $v$ the {\em ground state} for $w$.
A critical Hardy type weight $w$ is {\em positive critical} if the ground state $v$ belongs to $\ell^2(X,|w|m)$. For Hardy weights, this is equivalent to 
 $ v\in \mathcal{D}_0 $, see below. Note that  $(\mathcal{\mathcal{Q}}-w)(v)=0 $ in this case, i.e., the Hardy inequality allows for a minimizer for the quadratic form $\mathcal{Q}-w$. 
We say in analogy to  \cite{DFP14, KPP18}, that a Hardy type weight $w$ is called {\em optimal} if it is null-critical and {\em optimal near infinity}, which means that there is no $\lambda  > 0$ such that $\mathcal{Q} \geq (w+\lambda|w|)$ holds for all $\varphi\in C_c(X)$ supported on $X \setminus K$ for some finite $K \subseteq X$. 
By now, it has been observed that null-critical Hardy weights are automatically optimal at infinity, see \cite{KPP20,KN,Fischer,hake2025optimal,KovarikPinchover}. Since there is no published reference yet in the literature which includes our setting, we give a proof for the convenience of the reader.

\begin{proposition}[Optimality near infinity]\label{lem:optimalnearinfty}
	A null-critical Hardy type weight $w$ is optimal near infinity, i.e.,\@ for each $\lambda > 0$ and each finite set $K \subseteq X$, there is some $\varphi \in C_c(X)$ supported in $X \setminus K$ such that 
	\[
	\mathcal{Q}(\varphi) \, < \, \sum_{X \setminus K}m(w+\lambda |w|) \varphi^2.
	\] 
\end{proposition}

The proof of the proposition is based on the {\em ground state transform} involving the form $\mathcal{Q}_v$ defined for  $v \in \mathcal{F}$ as
\[
\mathcal{Q}_v(\varphi) = \frac{1}{2} \sum_{x,y \in X} b(x,y)v(x)v(y) \left( {\varphi(x)} - {\varphi(y)}  \right)^2.
%\, + \, \sum_{x \in X} \frac{\mathcal{L} v(x)}{v(x)} \varphi^2(x)m(x).
\] 

\begin{proposition}[Ground state transform,  Proposition~4.8 in \cite{KPP20}] \label{lem:GST}
	Let $v \in \mathcal{F}$ be strictly positive. Then for $\varphi \in C_c(X)$,
	\begin{align*}
		\mathcal{Q}(\varphi) = \mathcal{Q}_v\left( \frac{\varphi}{v} \right) \, + \, \sum_{ X} m\frac{\mathcal{L}v}{v}\varphi^2.
	\end{align*}
\end{proposition}

%We are now in position to prove optimality near infinity for null-critical Hardy type weights.  For the proof of Hardy type inequalities, we will also come back to the ground state transform below. 

\begin{proof}[Proof of Proposition~\ref{lem:optimalnearinfty}]
		Fix $\lambda > 0$ and let $K \subseteq X$ be finite. Let $v$ be the ground state of $w$. Observe that $w= \mathcal{L}v/v$ by Theorem~\ref{lem:characterisaton_criticality}. By the very same theorem, there is a sequence $(e_n)$  in $C_c(X)$ that converges pointwise to $v$ such that $0 \leq e_n \leq v$ and $\lim_{n \to \infty} (\mathcal{Q}-w)(e_n) = 0$.
		By the ground state transform, we have
		%\begin{align*} 			\mathcal{Q}_v \!\left( \frac{e_n}{v} \right) = (\mathcal{Q}-w)(e_n) \xrightarrow{n \to \infty} 0, 		\end{align*}
		and 
		\begin{align*}
			(\mathcal{Q}-w)(1_{X \backslash K} e_n)& = \mathcal{Q}_v \!\left( \frac{e_n - 1_{K} e_n}{v} \right) \leq 2 \mathcal{Q}_v \!\left( \frac{e_n}{v} \right) + 2 \mathcal{Q}_v \!\left( \frac{1_{K} e_n}{v} \right)\\
&\le 2(\mathcal{Q}-w)(e_n)+\sum_{x \in K} v(x) \sum_{y \in X} b(x,y) v(y)
,
		\end{align*}
		where we used the ground state transform and   $0 \leq e_n \leq v$ as well as $0 \leq \frac{1_K e_n}{v} \leq 1$ in the last estimate.
		Since $K$ is finite and $v \in \mathcal{F}$, it follows that % $\sup_{n \in \N} \mathcal{Q}_v\!\left( \frac{1_A e_n}{v} \right) < \infty$ and in turn that 
		$\sup_{n \in \N} (\mathcal{Q}-w)(1_{X \backslash K} e_n) < \infty$. On the other hand, it follows from Fatou's lemma %that for the quadratic form $|w|$, 
		\begin{align*}
			\liminf_{n \to \infty} \abs{w} \!(1_{X \backslash K} e_n)% &= \liminf_{n \to \infty} \sum_{x \in X \backslash A} \abs{w(x)} \! e_n^2(x)m(x) \\
			 &\geq \sum_{X \backslash K} m\abs{w} v^2.
		\end{align*}
		The sum on the right hand side is infinite since $w$ is null-critical and $K$ is finite.
		It follows for sufficiently large $n \in \N$ that $(\mathcal{Q}-w)(1_{X \backslash K} e_n) < \lambda \!\abs{w} \!(1_{X \backslash K} e_n)$. Therefore, $w$ is optimal near infinity. 
	\end{proof}
	
	We obtain the following representation of a critical Hardy weight via the Green kernel.

	\begin{proposition}
		\label{prop:ground_state_potential}
		Let $w\ge0$ be a critical and non-trivial Hardy weight. Then the ground state $v$ is a non-trivial superharmonic potential, i.e.,\@ there exists a nontrivial $0\le k\in \mathcal{G}$ such that
		\begin{align*}
			v=Gk\quad\mbox{and}\quad w=\frac{k}{Gk}.
		\end{align*}
	\end{proposition}
	
	The proof is based on the next lemma which is a variant of saying that ground states have minimal growth.

	%\Hmm{Say something about minimal growth here?}
	%and on the Riesz decomposition, see Theorem~\ref{thm:Riesz_decomposition}.
	
	\begin{lemma}
		\label{lem:critical_weights_subharmonic_function}
		Let $w$ be a critical and non-trivial Hardy weight with ground state $v$. If $u \geq 0$ is a subharmonic function such that $u/v$ is bounded, then $u = 0$.
	\end{lemma}

	\begin{proof}
		If ${u}/{v}$ is bounded, then there exists $C > 1$ such that $Cv - u\gneq0$. Further, we compute
		\begin{align*}
			(\mathcal{L} - w)(Cv - u) = C \underbrace{(\mathcal{L} - w)v}_{=0} - \underbrace{\mathcal{L} u}_{\leq 0} + \underbrace{wu}_{\geq 0} \geq 0.  
		\end{align*}
		By Theorem~\ref{lem:characterisaton_criticality}~$  (ii) $ this implies that there exists $a \in \R$ such that $Cv - u = av$. Since $u \geq 0$ and $v > 0$, we must have $C \geq a$ and because of
		\begin{align*}
			awv = \mathcal{L} av = \mathcal{L} (Cv - u) = C \mathcal{L} v - \mathcal{L} u \geq Cwv\gneq0,
		\end{align*}
		 we obtain $a \geq C$. We conclude $a = C$ and $u = 0$.
	\end{proof}

	\begin{proof}[Proof of Proposition~\ref{prop:ground_state_potential}]
		Let $w$ be a critical non-trivial Hardy weight with ground state $v$. Then $\mathcal{L} v = wv \gneq 0$ and so $v$ is superharmonic but not harmonic. By the Riesz decomposition, Theorem~\ref{thm:Riesz_decomposition}, there exists a potential $v_p \geq 0$ and a harmonic function $v_h \geq 0$ such that $v = v_p + v_h$. Then ${v_h}/{v} \leq 1$ and so $v_h = 0$ by the previous lemma (applied with $u=v_h$), and $v = v_p$ is a potential. Therefore, we get with Proposition~\ref{prop:characterisation_potentials} the representation $ v=Gk $
		%\begin{align*} 			v(x) = \sum_{y \in X} G(x,y) K(y) 		\end{align*}
		with $k = \mathcal{L} v \gneq 0$, which leads to $ w =\mathcal{L}v/v=k/Gk $.
	%	\begin{align*} 			w(x) = \frac{\mathcal{L} v(x)}{v(x)} = \frac{K(x)}{\displaystyle \sum_{y \in X} G(x,y) K(y)}. 		\end{align*} 
	\end{proof}

\subsection{Characterization of positive critical Hardy type weights}

In this section we will characterize positive critical Hardy type weights. %As before, we assume that $b$ is a connected graph over $(X,m)$. %The two following theorems can be found in the dissertation thesis of the first named author, cf.\@ \cite[Theorem~2.13, Theorem~2.15 and Theorem~2.16]{hake2025optimal} and are slightly extended here to arbitrary discrete measure spaces $(X,m)$ and Laplacians with possibly non-trivial killing term $c$.  

\begin{theorem}
	\label{thm:classification_positive-critical_weights}
	Let $w:X \rightarrow \R$. Then $w$ is a positive critical Hardy-type weight if and only if $w = {\mathcal{L} v}/{v}$ for some $v \in \D_0$, $v > 0$ and $\mathcal{L} v^2 \in \ell^1(X,m)$. 
\end{theorem}

%For the special situation of Hardy weights we are going to prove Theorem~\ref{thm:classification_positive-critical_weights_Hardy} stated in the introduction. \medskip

We start with a lemma that provides an extension of a Hardy type inequality to a larger space. 

\begin{lemma}
	\label{lem:extension_hardy_inequality}
	Let $w$ satisfy a Hardy-type inequality. Then 
	\begin{align*}
		\D_0 \cap \ell^2(X,\abs{w}m) = \D_0 \cap \ell^2(X, w_-m).
	\end{align*}
	Moreover, for every $\varphi \in \D_0 \cap \ell^2(X, w_- m)$, we have $\mathcal{Q}(\varphi) \geq w(\varphi)$ with equality if and only if $\varphi = 0$ or if there exists a ground state $v$ and a constant $c \in \R$ with $\varphi = cv$.
\end{lemma}

\begin{remark}
	If $w$ is a Hardy weight then $w_- = 0$ and so $\ell^2(X, w_-m) = C(X)$ and the lemma implies $\D_0 \subseteq \ell^2(X, w)$. Consequently, we can extend the Hardy inequality to all of $\D_0$. 
\end{remark}

\begin{proof}
	$ \subseteq $: Clearly $\ell^2(X,\abs{w}m) \subseteq \ell^2(X,w_-m)$ which gives the desired inclusion.
	%and so
	%\begin{align*} 		\D_0 \cap \ell^2(X,\abs{w}m) \subseteq \D_0 \cap \ell^2(X, w_-m). 	\end{align*}
	$ \supseteq $: Let $\varphi \in \D_0 \cap \ell^2(X, w_-m)$. Then, there exists a sequence $(\varphi_n)$ in $C_c(X)$ such that $\varphi_n \to \varphi$ pointwise, $\abs{\varphi_n} \leq \abs{\varphi}$ for all $n \in \N$ and such that $\mathcal{Q}(\varphi_n) \to \mathcal{Q}(\varphi)$, see \cite[Lemma~6.6]{KLW}.  
	Since $w$ satisfies a Hardy-type inequality, we have that
	\begin{align*}
		w_+(\varphi_n) \leq \mathcal{Q}(\varphi_n) +  w_-(\varphi_n)
	\end{align*}
	for all $n \in \N$. From the dominated convergence theorem, we get that $w_-(\varphi_n) \to w_-(\varphi)$. 
	Fatou's lemma implies that
	\begin{align*}
		w_+(\varphi) \leq \liminf_{n \to \infty} w_+(\varphi_n) \leq \liminf_{n \to \infty} \left( \mathcal{Q}(\varphi_n) +  w_-(\varphi_n) \right) = \mathcal{Q}(\varphi) + w_-(\varphi) < \infty.
	\end{align*}
	This yields that $\varphi \in \ell^2(X, w_+ m)$ and so $\varphi \in \D_0 \cap \ell^2(X,\abs{w}m)$ and
	\begin{align*}
		w(\varphi) = w_+(\varphi) - w_-(\varphi) \leq \mathcal{Q}(\varphi).
	\end{align*}
	Furthermore, with $w_+(\varphi) < \infty$, the dominated convergence theorem implies $w_+(\varphi) = \lim_{n \to \infty} w_+(\varphi_n)$, and so $(\mathcal{Q}-w)(\varphi) = \lim_{n \to \infty} (\mathcal{Q}-w)(\varphi_n)$. If $\varphi$ is a ground state, then $(\mathcal{Q}-w)(\varphi_n) \to 0$, see Theorem~\ref{lem:characterisaton_criticality}, and so then $\mathcal{Q}(\varphi) = w(\varphi)$. This immediately implies that this equality also holds if $\varphi = cv$ for any ground state $v$ and $c \in \R$. 
	
	Finally, assume that $\mathcal{Q}(\varphi) = w(\varphi)$ and $\varphi \neq 0$. %We have to show that $\varphi$ is necessarily of the form $\varphi = cv$ for some ground state $v$ and a constant $c \in \R$. 
	Pick $o \in X$ with $\varphi(o) \neq 0$. Then $\varphi_n(o) \to \varphi(o) \neq 0$ and $(\mathcal{Q}-w)(\varphi_n) \to (\mathcal{Q}-w)(\varphi) = 0$ and so $(\varphi_n)$ satisfies all the conditions from Theorem~\ref{lem:characterisaton_criticality} $(iii)$. Therefore $w$ must be critical and there exists a ground state $v \in \F$. By the ground state transform, cf.\@  Lemma~\ref{lem:GST},  
	\begin{align*} % \label{eqn:GST}
		\mathcal{Q}_v\left(\frac{\varphi_n}{v}\right) = (\mathcal{Q}-w)(\varphi_n),
	\end{align*}
	and this converges to $0$ as $n \to \infty$. 
	This implies that $(\frac{\varphi_n}{v})$ converges pointwise to a constant function. Therefore, $\varphi = cv$ for some $c \in \R$ by connectedness of the graph.
\end{proof}

We arrive at a first characterization of positive criticality.

\begin{proposition}
	\label{prop:characterisation_positive_criticality}
	Let $w$ satisfy a Hardy-type inequality. Then the following assertions are equivalent:
	\begin{enumerate}[(i)]
		\item $w$ is positive-critical.
		\item $w$ is critical and its ground state is in $\D_0 \cap \ell^2(X,w_-m)$. 
		\item There exists a non-trivial $\varphi \in \D_0 \cap \ell^2(X, \abs{w}m)$ with $Q(\varphi) = w(\varphi)$.
	\end{enumerate}
\end{proposition}

\begin{remark}
	% For $c=0$ and $m=1$, this proposition can also be found in \cite[Proposition~2.18]{hake2025optimal}.
	Again, for a Hardy weight, we have $\D_0 \cap \ell^2(X, w_-m) = \D_0$. Consequently, a critical Hardy weight is positive-critical if and only if the ground state is in $\D_0$. This equivalence can also be found in \cite[Theorem~6.2]{KPP20}.
\end{remark}

\begin{proof}
	The equivalence of $(ii)\Longleftrightarrow(iii)$ and the implication $(ii) \Longrightarrow (i)$ follow immediately from the previous lemma and the definition of positive criticality. To prove $(i) \Longrightarrow (ii)$, let $w$ be positive critical with ground state $v$. By the definition of positive criticality, $w$ is critical and the ground state $ v $ is in $\ell^2(X, \abs{w}m) \subseteq \ell^2(X, w_-m)$. Hence, it suffices to prove that $v \in \D_0$. For that purpose, let $(e_n)$ be a sequence in $C_c(X)$ that converges pointwise to $v$, satisfies $0 \leq e_n \leq v$ and $\lim_{n \to \infty} (\mathcal{Q}-w)(e_n) = 0$. Such a sequence exists by Theorem~\ref{lem:characterisaton_criticality}. Then, by the dominated convergence theorem, we have $w(e_n) \to w(v)$ since $v \in \ell^2(X, \abs{w}m)$. This implies that $(\mathcal{Q}(e_n))$ must be convergent as well and, thus, $v \in \D_0$ by Lemma~\ref{prop:charD_0}.
\end{proof}

The following proposition will give further insight into the integrability properties of a ground state. %For $m=1$ and $c=0$, see also \cite[Proposition~2.20]{hake2025optimal}.

\begin{proposition}
	\label{prop:v_in_D_0_summabilities}
	Let $v \in \D_0$, $v > 0$ and $w = {\mathcal{L} v}/{v}$. Then the following assertions are equivalent:
	\begin{enumerate}[(i)]
		\item $\mathcal{L} v^2 \in \ell^1(X,m)$.
		\item $v \in \ell^2(X, \abs{w}m)$.
		\item $v \in \ell^2(X, w_-m)$.
	\end{enumerate}
	If these  conditions hold then $\mathcal{Q}(v) = w(v)$ and $$\sum_{ X}m \mathcal{L} v^2  = \sum_{ X} cv^2.$$
\end{proposition}

\begin{proof}
	Let $v \in \D_0$, $v > 0$ and $w = {\mathcal{L} v}/{v}$. From the ground state transform,  Lemma~\ref{lem:GST}, we know that $w$ satisfies a Hardy-type inequality. Moreover, we have 
		$ \D_0 \cap \ell^2(X,\abs{w}m) = \D_0 \cap \ell^2(X, w_-m) $
	by Lemma~\ref{lem:extension_hardy_inequality}. This immediately implies $(ii) \Longleftrightarrow (iii)$. Note that $v,v^2 \in \F$ since $v \in \D_0 \subseteq \D$. A straight forward computation, cf.~\cite[Lemma 2.1]{KPP18} yields  
	\begin{align*}
		\mathcal{L} v^2(x)m(x) = 2 v(x) \mathcal{L} v(x) m(x) - \sum_{y \in X} b(x,y) (v(x)-v(y))^2- c(x)v^2(x)
	\end{align*}
	for all $x \in X$. The fact  $v \in \D$ %means
	%\begin{align*}		\sum_{x \in X} \sum_{y \in X} b(x,y) (v(x)-v(y))^2 + 2 \sum_{x \in X} c(x)v(x)^2 = 2\mathcal{Q}(v)  < \infty. 	\end{align*} This 
	implies 
	that $\sum_{ X} m\mathcal{L} v^2$ is absolutely convergent if and only if $$ \sum_{ X} m v \mathcal{L} v = \sum_{X} mwv^{2}$$ is. This is exactly the equivalence $(i) \Longleftrightarrow (ii)$.
	
	Now, let us assume that the equivalence holds.
	We calculate using $ w=\mathcal{L}v/v $ and a sequence $(v_n)$ in $C_{c}(X) $ with $|v_n| \leq v$ which converges to $ v $ with respect to $ \mathcal{Q} $ (see \cite[Lemma~6.6]{KLW}) 
	\begin{align*}
		\sum_{X}m wv^{2} = \sum_{ X}m v \mathcal{L} v=\lim_{n\to\infty} \sum_{ X}m v_{n} \mathcal{L} v =\lim_{n\to\infty}\mathcal{Q}(v_{n},v) = \mathcal{Q}(v),
	\end{align*}
where the  second equality follows by dominated convergence since $ \mathcal{L}v =wv$ and $ v\in \ell^{2}(X,|w|m) $ and the third inequality follows from Green's formula  \cite[Proposition~1.5]{KLW}. Furthermore, by the  considerations above, we have
	\begin{align*}
		\sum_{ X}m \mathcal{L} v^2 &= 2 \sum_{X}m v \mathcal{L} v- 2\mathcal{Q}(v)+\sum_{X}cv^{2}= \sum_{X}cv^{2}. 
	\end{align*}
	This finishes the proof.
\end{proof}

\begin{corollary}
	\label{cor:superharmonic_implies_summability_condition}
	If $v \in \D_0$, $v > 0$ is superharmonic, then $\mathcal{L} v^2 \in \ell^1(X,m)$ and $$\displaystyle \sum_{  X}m \mathcal{L} v^2  = \sum_{ X} c v^2 .$$
\end{corollary}

\begin{proof}
	Since $v > 0$ is superharmonic,  the condition $(iii)$ in Proposition~\ref{prop:v_in_D_0_summabilities} is trivially satisfied for $w= {\mathcal{L} v}/{v}$,
	because $w_- = 0$. 
\end{proof}

We have now everything in place for  the proof the characterization of positive critical Hardy-type weights in Theorem~\ref{thm:classification_positive-critical_weights}.

\begin{proof}[Proof of Theorem~\ref{thm:classification_positive-critical_weights}]
	Since $w$ is a positive critical Hardy-type weight,  by Proposition~\ref{prop:characterisation_positive_criticality} its ground state $v$ is in $\D_0 \cap \ell^2(X,w_-)$. By Proposition~\ref{prop:v_in_D_0_summabilities}, this gives  $\mathcal{L} v^2 \in \ell^1(X,m)$ and the ground state is strictly positive. 
	
	For the converse, let $v \in \D_0$, $v > 0$, $\mathcal{L} v^2 \in \ell^1(X,m)$ and $w = {\mathcal{L} v}/{v}$. Then $w$ satisfies a Hardy-type inequality by the ground state transform,  Lemma~\ref{lem:GST}. Further, we infer from Proposition~\ref{prop:v_in_D_0_summabilities}  that $v \in \D_0 \cap \ell^2(X, \abs{w}m)$ and $\mathcal{Q}(v) = w(v)$ and so $w$ is positive critical by Proposition~\ref{prop:characterisation_positive_criticality}.
\end{proof}

\begin{proof}[Proof of Theorem~\ref{thm:classification_positive-critical_weights_Hardy}]
Given a critical Hardy weight $ w $ with ground state,	it follows from the characterization of critical Hardy weights, Theorem~\ref{lem:characterisaton_criticality}~$(iv)$, that $w=\mathcal{L} v / v$. Moreover, if $w$ is positive critical, then $v \in \D_0$ by Proposition~\ref{prop:characterisation_positive_criticality} and $\mathcal{L} v^2 \in \ell^1(X,m)$ by Proposition~\ref{prop:v_in_D_0_summabilities}. This shows $(i)\Longrightarrow(ii)$. The converse is immediate from Proposition~\ref{prop:v_in_D_0_summabilities} and Theorem~\ref{lem:characterisaton_criticality}.
	Since $\ell^2(X,w_{-}m) =C(X)$,
	the equivalence $(ii) \Longleftrightarrow (iii)$ follows from Proposition~\ref{prop:v_in_D_0_summabilities}. As for $(iii) \Longrightarrow (iv)$, note that we observe that a superharmonic $v \in \D_0$ is necessarily a potential with $\mathcal{L} v \in \mathcal{G}_2$
	by Proposition~\ref{prop:equivalence_P_D_0} which is exactly of the form $v = Gk$ with $k = \mathcal{L} v \gneq 0$. For the converse $(iv) \Longrightarrow (iii)$, the very same Proposition~\ref{prop:equivalence_P_D_0} yields $v=Gk \in \D_0$ and $\mathcal{L} v = k \gneq 0$. Clearly, $v$ is strictly positive. 
\end{proof}

%INSERT LATER:
%A Hardy type weight $w$ is called \emph{critical} if for every Hardy weight $w'\ge w$ we have $w'=w$. A critical Hardy type weight $ w $ is called \emph{positive critical} if  there exists a  function $ v\in \mathcal{D}_0 $ such that $(\mathcal{\mathcal{Q}}-w)(v)=0 $, i.e., the Hardy inequality allows for a minimizer in $\mathcal{D}_0$.

\section{Fractional  Laplacians on graphs} \label{sec:fractional}

In this section, we define the fractional Laplacian over a connected  graph $(b,c)$ over $(X,m)$. Further, we prove the characterization of positive criticality stated in the introduction as Theorem~\ref{t:poscritical}.

\subsection{The fractional Laplacian}
Given the self-adjoint operator $ L $ associated to the closure $ Q$ of the restriction of $ \mathcal{Q} $ to $ C_{c}(X) $ on $ \ell^{2} (X,m)$, by the spectral theorem, the fractional Laplacian $L^\sigma$, $\sigma\in (0,1)$, can be represented via the semigroup for $f\in D(L^\sigma)$ as
\begin{align*}
    L^{\sigma}f=\frac{1}{|\Gamma(-\sigma)|}\int_{0}^{\infty}(I-e^{-tL})f\frac{d t}{t^{1+\sigma}}.
\end{align*}
This operator comes with a quadratic form $Q^\sigma$ defined on $D(Q^\sigma)=D(L^{\sigma/2})$ by $Q^\sigma(f)=\langle L^{\sigma/2}f,L^{\sigma/2}f\rangle$.

We show below that $L^\sigma$ is associated with a connected graph $(b_\sigma,c_\sigma)$ over $(X,m)$ for $ \sigma\in (0,1) $. 
Since $L$ is the generator of a Markovian semigroup, cf. \cite[Chapter~1]{KLW}, the semigroup extends to $ \ell^{\infty}(X) $. We  define the killing term
\[c_\sigma(x)=\frac{1}{|\Gamma(-\sigma)|}\int_0^\infty m(x)\bigl(1-q_t(x)\bigr)\,
\frac{dt}{t^{1+\sigma}}, \]
for $ x\in X $, 
where 
$$q_t(x) = e^{-tL}1(x) = {\sum_{y\in X} p_t(x,y)m( y)},\quad
p_t(x,y) = m(y)^{-1} e^{-tL}1_y(x).$$
 Since, $0\le e^{-tL}1\leq 1 $, the function $ q_t $ takes values in $[0,1]$ and, therefore, $c_\sigma\ge 0$. Furthermore, we define the edge weights 
 $b_\sigma:X\times X\to [0,\infty] $  by $b_\sigma(x,x)=0$  and otherwise
\begin{align*}% \label{eqn:balpha}
 b_{\sigma}(x,y)
 %&=\frac{1}{|\Gamma(-\al)|}\int_{0}^{\infty}m(x)e^{-tL}1_y(x)\frac{\d t}{t^{1+\al}} \nonumber \\
 &=\frac{1}{|\Gamma(-\sigma)|}\int_{0}^{\infty}m(x)m(y)p_t(x,y)\frac{\d t}{t^{1+\sigma}}, 
\end{align*}
 where  strict positivity of $ p_t $, \cite[Theorem~1.26]{KLW}, implies the strict positivity of $ b_{\sigma} $ for all $ x,y\in X $ whenever $ b_{\sigma}(x,y) $ is finite. Note that $b_\sigma$ is symmetric since $p_t$ is symmetric. The following theorem is a generalization of \cite[Theorem 2.1]{ZLY} to the setting of general graphs which are allowed to be non-locally finite or stochastically incomplete. Recall that a quadratic form $Q$ is called a {\em Dirichlet form} if $c \circ f \in D(Q)$ and $Q(c \circ f ) \leq Q(f)$ for all normal contractions $c:\R \to \R$. For graphs, a Dirichlet form is {\em regular} if $C_c(X)$ is dense in $D(Q)$. For more background information on (regular) Dirichlet forms, we refer to \cite[Chapter~1]{KLW}.

\begin{theorem}\label{t:fractionalgraph}
    Let $(b,c)$ be a connected graph over $ (X,m) $ and $\sigma\in (0,1)$.  Then,  $  D(Q)\subseteq  D(Q^{\sigma}) $ and $  D(L)\subseteq  D(L^{\sigma}) $, and we have for $f\in D(Q^{\sigma})$
    \begin{align*}
        Q^{\sigma}(f)=\frac{1}{2}\sum_{x,y\in X}b_\s(x,y)(f(x)-f(y))^2+\sum_{x\in X} c_\s(x) f(x)^2.
    \end{align*}
In particular, $ Q^{\sigma} $ is a regular Dirichlet form and $(b_\sigma,c_\sigma)$ is a graph over $ (X,m) $ and $L^\sigma$ is a restriction of
    \begin{align*}
        \L^{\sigma}f(x)=\frac{1}{m(x)}\sum_{y\in X}b_{\sigma}(x,y)(f(x)-f(y))+\frac{c_\sigma(x)}{m(x)} f(x)
    \end{align*}
    on 
    \begin{align*}
        \F^\sigma=\{f:X\to \R\mid \sum_{y\in X}b_\sigma(x,y)|f(y)|<\infty\mbox{ for all }x\in X\}.  
    \end{align*}
\end{theorem}
\begin{proof}
    The statements about the domains follow from the spectral theorem since $ \lambda\mapsto \lambda^\sigma $ for $ \sigma\in (0,1) $ is integrable with respect to a finite measure on $ (0,\infty) $  whenever  $ \lambda\mapsto \lambda$ is. Furthermore, by spectral calculus we have for $ f\in D(L^{\sigma/2}) $
    \begin{align*}
        Q^{\sigma}(f)=\int_0^1 \lambda^{\sigma} d\mu_f  +\int_1^\infty \lambda^{\sigma} d\mu_f  \leq \mu_f([0,1]) +\int_1^\infty \lambda d\mu_f  \leq \|f\|^2 + Q(f).
    \end{align*}
    Thus, density of $ C_c(X) $ in $ D(Q) $ implies density in $ D(Q^{\sigma}) $ which implies regularity. For \(f\in D(L^{\sigma})\),
\[
\langle f,L^\sigma f\rangle
=\frac{1}{|\Gamma(-\sigma)|}\int_0^\infty
\langle f,(I-e^{-tL})f\rangle\,\frac{dt}{t^{1+\sigma}} .
\]

Since the semigroup is Markovian,  $f\mapsto \langle f,(I-e^{-tL})f\rangle$ on $D(L^\sigma)$ is Markovian. Thus, the form of $L^\sigma$ is a regular Dirichlet form.

Now, we turn to the representation of the form. Using \(p_t(x,y)=p_t(y,x)\) we obtain
\begin{multline*}
\langle f,(I-e^{-tL})f\rangle
=\sum_{x\in X} f(x)^2 m(x)-\sum_{x,y\in X} f(x)f(y)\,p_t(x,y)\,m(x)m(y) \\
=\frac12\sum_{x,y\in X}\bigl(f(x)-f(y)\bigr)^2 p_t(x,y)\,m( x)m( y)  +\sum_{x\in X} f(x)^2\big( 1- q_t(x) \big) m( x) .
\end{multline*}
By Tonelli's theorem (all integrands are non‑negative), we interchange the sum and the integral
\begin{align*}
\langle f,L^\sigma f\rangle
&=\frac12\sum_{x,y\in X}\bigl(f(x)-f(y)\bigr)^2\,
\Bigg(\frac{1}{|\Gamma(-\sigma)|}
\int_0^\infty p_t(x,y)\frac{dt}{t^{1+\sigma}}\Bigg)m(x)m(y) \\
&\quad +\sum_{x\in X} f(x)^2\,
\Bigg( \frac{1}{|\Gamma(-\sigma)|}
\int_0^\infty\bigl(1-q_t(x)\bigr)\frac{dt}{t^{1+\sigma}}\Bigg)m(x) .
\end{align*}
This gives the desired representation on $D(L^\s)$ and therefore on $D(L^{\s/2})$ by density. Since $Q^\s$ is a regular Dirichlet form, it can be represented by a graph $(b,c)$ over $ (X,m) $ cf. \cite[Theorem 1.18]{KLW} which is given by $(b_\sigma,c_\sigma)$ as calculated above.
The statement about the operator now follows from \cite[Theorem 1.6]{KLW}.
This finishes the proof.
\end{proof}

The theorem above gives in particular finiteness of $b_\s$ for $\s\in (0,1)$.

The Green's function $G^\sigma:X\times X\to(0,\infty]$ of the fractional Laplacian $ L^{\sigma} $, $ \sigma>0 $, is given as
\begin{align*}
	G^\sigma(x,y)=\lim_{E\searrow 0} (L^{\sigma}+E)^{-1}1_y(x)=\int_{0}^{\infty} e^{-tL^\sigma}1_{y}(x){\d t}
\end{align*}
see \cite[Theorem~6.26]{KLW}. We next show  that $ G^{\sigma} $ 
can be represented via a  kernel which can be understood as $ b_{-\sigma} $.

\begin{theorem}\label{p:Green}  For all $ \s\in (0,1) $, we have
    \begin{align*}
        G^\sigma(x,y) =\frac{1}{|\Gamma(\sigma)|}\int_{0}^{\infty}m(y)p_t(x,y)\frac{\d t}{t^{1-\sigma}} .
    \end{align*}
\end{theorem}
\begin{proof}
Let $\mu = \mu_{h}$ be the  real-valued spectral measure for $L$ with respect to $h \in \ell^2(X,m)$. Since $({\lambda^\sigma + E} )^{-1}, ({\lambda + E} )^{-\sigma}\nearrow {\lambda^{-\sigma}}$ as $E \searrow 0$ for $\lambda > 0$,  we have by  the spectral theorem and the monotone convergence theorem,
\begin{multline*}
\lim_{E \searrow 0} \langle (L^\sigma + E)^{-1} h,h  \rangle=\lim_{E \searrow 0}\int_0^\infty ({\lambda^\sigma + E})^{-1} \, d\mu(\lambda) = \int_0^\infty \lambda^{-\sigma} \, d\mu(\lambda)	\\
=\lim_{E \searrow 0}\int_0^\infty ({\lambda + E})^{-\sigma} \, d\mu(\lambda)
= \lim_{E \searrow 0}\langle (L + E)^{-\sigma} h, h \rangle,
\end{multline*}
where both sides can be infinite.
Letting $h=e_x = 1_x/\sqrt{m}$ for some $x \in X$, we have
\begin{multline*}
	G^{\sigma}(x,x)=\lim_{E \searrow 0}\langle(L^{\sigma}+E)^{-1}e_x,e_x\rangle=	\lim_{E \searrow 0}\langle(L+E)^{-\sigma}e_x,e_x\rangle\\=	\lim_{E \searrow 0}(L+E)^{-\sigma}1_x(x)
	=\lim_{E \searrow 0}\frac{1}{|\Gamma(\s)|}\int_{0}^{\infty}e^{-t(L+E)}1_x(x)\frac{\d t}{t^{1-\sigma}}\\
	=\frac{1}{|\Gamma(\s)|}\int_{0}^{\infty}e^{-tL}1_x(x)\frac{\d t}{t^{1-\sigma}},
\end{multline*}
where the last equality follows by monotone convergence which is applicable since the semigroup is positivity preserving. The statement for $x\neq y$ follows by polarization and taking into account that $G^\sigma> 0$ and $G^\sigma(x,y)<\infty$ if and only if $G^\sigma(x,x)<\infty$, \cite[Theorem~6.26]{KLW}.% whenever $x,y \in X$ are in the same connected component.
\end{proof}

\subsection{Hardy weights for the fractional Laplacian}

%In this section, we restrict ourselves for simplicity to connected graphs. The general case can be recovered by considering the connected components separately. 
For  $o\in X $ and $\alpha>0$, we recall the function
\begin{align*}
    k_{\alpha}(x)=k_{\alpha,o}(x)=\frac{1}{|\Gamma(\alpha)|}\int_{0}^{\infty}m(o)p_t(x,o) \frac{\d t}{t^{1-\alpha}}, \quad x \in X,
\end{align*}
and observe that $ k_{\alpha} $ is strictly positive by strict positivity of $ p_t $ and for $ \alpha\in (0,1) $, we have $ k_{\alpha,o}=G^{\alpha}(\cdot,o) $.
For $\alpha = 0$, we set $k_0 = k_{0,o} = 1_o$. 
 Furthermore, we  defined 
\begin{align*}
    \alpha_{\max}=\sup\{\alpha>0\mid k_{\alpha,o}(o)<\infty \}.
\end{align*}
%Next, we show that $ \alpha_{\max} $ is independent of the choice of $ o\in X $ and that $ k_{\alpha} $ is finite for all $\alpha\in (0,\alpha_{\max})$.
 
\begin{lemma}\label{lem:amax}
%	Let $(b,c)$ be a connected graph over $ (X,m) $ and $ o\in X $.
	The number $ \alpha_{\max} $ is independent of $o\in X$
and  $k_\alpha=k_{\alpha,o}$ is finite for all $\alpha\in (0,\alpha_{\max})$
\end{lemma}
\begin{proof}
	 The convergence of the integral defining $k_{\alpha}$ depends only at the behavior at infinity for $\alpha > 0$ since $ m(o)p_t(x,o)\le 1/m(x) $ and $ t\mapsto t^{1-\alpha} $ is integrable at $ 0 $. This gives finiteness for all $\beta<\alpha$ and fixed $x,y\in X$. For $x,y\in X$, one readily checks that for $t>1$ 
    \begin{align*}
      p_t(x,o) &\geq  m(o)^{-1} e^{-(t-1)L} 1_o(y) e^{-L}1_y(x) =C_{x,y}p_{t-1}(y,o)
    \end{align*} cf.~the proof of \cite[Theorem~6.26~(c)]{KLW}. 
	   Since the convergence of the integral depends only at the behavior at infinity, we infer that if $k_{\alpha}(x)$ is finite for some $ x\in X $, then $k_{\alpha}(y)$ takes finite values for all $ y\in X $. Using symmetry of $ p_{t} $, we get also the independence of $ o $ in regard to the finiteness.
\end{proof}

Next, we show a key result which shows that the functions $ k_\alpha $ are superharmonic potentials for the fractional Laplacian and give rise to Hardy weights.

\begin{proposition}\label{p:exponents}
	Let %$(b,c)$ be a connected graph over $ (X,m) $ and 
	$ \sigma\in (0,1] $ and 
	$\alpha\in [\sigma,\alpha_{\max})$. Then, $k_\alpha\in \F^\sigma$ and
	\begin{align*}
		\L^{\sigma}k_{\alpha}=k_{\alpha-\sigma}.
	\end{align*} 
	Furthermore, %if $k_\sigma=G^\sigma(\cdot,\,o)$ is  finite, then 
	$k_{\alpha-\sigma}\in \G^\sigma$  and
	\begin{align*}		k_{\alpha}= G^\sigma k_{\alpha-\sigma}.	
    \end{align*} 
    	In particular, $ k_{\alpha} $ is a superharmonic potential and  
	\begin{align*}
		w_{\sigma,\al}=\frac{\L^\sigma k_{ \alpha} }{k_{\alpha} }=\frac{k_{\al-\sigma}}{k_{\al}}=\frac{k_{\al-\sigma} }{G^\sigma k_{\al-\sigma}}
	\end{align*}
	is a non-negative weight which is strictly positive if $\alpha > \sigma$.
\end{proposition}
  \begin{proof}We define $b_{\beta,\ep} : X\times X\to [0,\infty]$ for $\beta\in \R\setminus \N_0$ and $\ep>0$ as 
		\[b_{\beta,\eps} (x,y)
		= \frac1{|\Gamma(-\beta)|} \int_0^\infty 
		m(x)e^{-t\ep} e^{-tL}1_{y}(x) \frac{\d t}{t^{1+\beta}} \] 
	 and $ b_{0,\ep}(x,y)=1_{y}(x) $ for $ \beta=0 $. Correspondingly, we define $$  k_{\beta,\ep}=\frac{1}{m}b_{-\beta,\ep}(\cdot,o)  $$ and for $\sigma\in (0,1)$ the killing term 
\[c_{\sigma,\ep}(x)=\frac{1}{|\Gamma(-\sigma)|}\int_0^\infty\bigl(1-m(x)e^{-t\ep} e^{-tL}1_y(x)\bigr)\,
\frac{dt}{t^{1+\sigma}}. \]
For  $ \beta >0 $,  it follows from the spectral theorem that	
		\begin{align*}
			 (L+\ep)^{-\beta} 1_{o}(x)= \frac{1}{m(x)}  b_{-\beta,\eps} (x,o) = k_{\beta,\varepsilon}(x),	
		\end{align*}
	and therefore, $ k_{\beta,\eps}\in D((L+\eps)^{\gamma})$ for $\gamma<\beta$. 
 Moreover, by Theorem~\ref{t:fractionalgraph} and its proof, $(L+\ep)^{\sigma}$ is associated with the graph $(b_{\sigma,\ep},c_{\sigma,\ep})$ over $(X,m)$.

Let $0<\s\leq\alpha<\alpha_{\max}$. By the spectral theorem, we have for $\ep>0$ that $(L+\ep)^{-\alpha}$ maps into $D((L+\ep)^{\sigma})$ and  therefore, $ k_{\alpha,\ep}\in D((L+\ep)^{\sigma}) $. Hence, by the  spectral calculus
	\begin{align*}
			 (L+\ep)^\s k_{\alpha,\eps} 
		= (L+\ep)^\s (L+\ep)^{-\alpha} 1_{o}
		= (L+\ep)^{-(\alpha-\s)} 1_{o}
		=k_{\al-\s,\eps}.
		\end{align*}
		Observe that	$k_{\beta,\eps}(x)$, $b_{\sigma,\eps}(x,y)$, $c_{\sigma,\eps}(x) 	 $ converge monotonously to $ k_{\beta} (x)$, $b_{\sigma}(x,y)$, $c_{\sigma}(x) 	 $ for $ x,y\in X $, as $ \eps\to0 $. Hence, the right hand side converges pointwise to $ k_{\alpha-\sigma} $ as $ \eps\to0 $.
To see that the left hand side converges pointwise to $ \L^{\sigma}k_{\alpha} $, we use the representation of $ (L+\eps)^{\sigma} $ and split up the left hand side into two terms as follows
    \begin{multline*}
        (L+\eps)^{\sigma}k_{\alpha,\eps}(x)\\
		=\frac{k_{\alpha,\eps}(x)}{m(x)}\left(\sum_{y\in X}b_{\sigma,\eps}(x,y) +c_{\sigma,\eps}(x)   \right)  -\frac{1}{m(x)}\sum_{y\in X}b_{\sigma,\eps}(x,y)k_{\alpha,\eps}(y).
    \end{multline*}
    The first term converges to $ \frac{k_{\alpha}}{m} (x)(\sum_y b_{\sigma}(x,y) + c_{\sigma}(x) )$ as $ \eps\to0 $. Since the right hand side is finite, we know that the limit of the second term exists. By monotone convergence, we conclude that this limit is given by
    \(\frac{1}{m(x)}\sum_{y\in X}b_{\sigma}(x,y)k_{\alpha}(y).\)
    Thus, we obtain that the left hand side converges pointwise to $ \L^{\sigma}k_{\alpha} (x) $ as $ \eps\to0 $.
    This shows the first statement.

The second statement follows similarly using spectral calculus
\begin{align*}
    k_{\alpha,\eps} &= (L+\eps)^{-\sigma}(L+\eps)^{-(\alpha-\sigma)}1_{o} = \frac{1}{|\Gamma(\sigma)|}\int_{0}^{\infty}\,e^{-t(L+\eps)}k_{\alpha-\sigma,\eps} \frac{\d t}{t^{1-\sigma}}.
\end{align*}
Now, letting $ \eps\to 0 $ and using monotone convergence  yields the claim for $\sigma \in (0,1)$ by the representation of the Green's function shown above in Theorem~\ref{p:Green}. For $\sigma=1$, it follows also from the spectral theorem that 
$(L+\varepsilon) k_{\alpha,\varepsilon} = k_{\alpha -1, \varepsilon}$ and it is clear that the left hand side converges pointwise to $\mathcal{L}k_{\alpha}$ as $\varepsilon \to 0$. This finishes the proof. 
\end{proof}

We can now prove our second main result, Theorem~\ref{t:poscritical}, stated in the introduction which in particular characterizes positive criticality of the Hardy weights $ w_{\sigma,\alpha} $.

\begin{proof}[Proof of Theorem~\ref{t:poscritical}]
By the proposition above, $w_{\sigma,\al}={k_{\al-\sigma} }/{G^\sigma k_{\al-\sigma}}$ is a Hardy weight. By Theorem~\ref{thm:classification_positive-critical_weights_Hardy}, $w_{\sigma,\al}$ is positive critical if and only if $$\sum_X k_{\al-\sigma}  k_{\al} \, m =\sum_X k_{\al-\sigma} G^\sigma k_{\al-\sigma} \, m <\infty,  $$ where the first equality also follows from the proposition above.
\end{proof}

\subsection{Asymptotics of the fractional Hardy weights} \label{sec:asymp}

For graphs endowed with an intrinsic metric with jump size $1$, we derive explicit asymptotics for the fractional Hardy weights assuming suitable heat kernel estimates and a growth condition on metric balls. We will prove Theorem~\ref{thm:MAIN2} stated in the introduction, along with stronger variants under the assumption of heat kernel asymptotics.   

\medskip

To get any meaningful Hardy inequality from the proposition above, we need to understand the asymptotic behavior of $k_\alpha$.  To this end, we need to assume bounds on the heat kernel.

% For the standard Laplacian on $\Z^d$, these are well-known, cf., e.g., \cite{LawlerLimic}.  
Let $\varrho$ be an intrinsic metric with jump size $1$ on the graph. For bounded graphs, this will typically be the combinatorial distance. For unbounded graphs, one does indeed need instrinsic metrics, cf. \cite{BHY17,Folz11,KR24}.

With regards to the heat kernel, for discrete spaces there are no Gaussian bounds for all times. The Gaussian bounds are only known to hold for large times,  and fail for small times, cf. \cite{KLMST}. So for large times, we consider more general bounds than Gaussian which include also sub-Gaussian bounds as they typically arise for fractal type  \eqref{e:GB}. These examples are discussed in more detail in Section~\ref{sec:examples}.
For the short time bounds, there are no Gaussian bounds, but instead one has universal short time bounds given by Davies-Gaffney type estimates, cf. \cite{BHY17}, which are sufficient for our purposes. In order to employ the Davies-Gaffney type estimate, we need a mild lower bound on the measure. Precisely, we say that  $m$ is {\em sub-exponentially bounded from below} if
	\begin{align*}
	 \liminf_{|x| \to \infty}  \frac{1}{|x|} \log m(x) \geq 0,
\end{align*}where $ 
|x|=\rho(x,o) $.

\begin{theorem}[Asymptotics of the Riesz kernel]\label{theorem:RieszAsymptotics}
Assume $m$ being sub-exponentially bounded from below and \eqref{e:GB} holds for $o\in X$, $d\ge 1$, $\beta\ge 2$. Then, for $ \al>0 $,  the Riesz kernel $k_{\alpha}$ is finite if and only if   $\alpha<d/\beta$. In this case,  
    \[  k_\alpha(x ) \asymp \, |x|^{-d+\alpha\beta}, 
\]
where the two-sided constant can be chosen uniform over all $\alpha$ taken from a compact subset of $(0,d/\beta)$. 
%In particular, we have $d/2 \leq \alpha_{\max}$ in both cases. 
Furthermore,  for every compact subset $D$ of $(0,d/\beta)$ and each fixed $x \in X$, the map $D\rightarrow (0,\infty),$  $\alpha\mapsto k_\alpha(x)$ extends to an analytic function  in a complex neighborhood of $D$.
\end{theorem}

%\begin{remark}	For the standard lattice graph over $\Z^d$, the asymptotic of the Riesz kernel as in~(a) holds with $\mu=2,\nu=1, \theta=d/2, a=4, A = (4\pi)^{-d/2}$ and $\delta=1$, even over the slightly larger range $\alpha \in (-1,d/2)$, see~ \cite[Theorem~A.1]{HKP}	(Note the different conventions: we have $k_{\alpha}$ = $\kappa_{-\alpha}$ with $\kappa$ as defined in the just cited paper.)\end{remark}

We start with a proposition which gives the necessary short time bounds for the heat kernel in the integral which defines the Riesz kernel. The proof relies on Davies-Gaffney type estimates as proved in \cite{BHY17}.

\begin{proposition}[Short time upper bounds for the heat kernel] \label{prop:SU}
	Assume that 
	$m$ is sub-exponentially bounded from below. For every bounded interval $I \subseteq (0,\infty)$, there are $c,C > 0$ such that for every $\alpha \in I$ and each $x \in X$, we have 
	\begin{align*}
		\int_0^{|x|} p_t(x,o) \frac{dt}{t^{1-\alpha}} \, \leq\, \frac{C}{\alpha}e^{-c|x|}% \exp \left( -\frac12|x| \log |x| \right).
	\end{align*}
\end{proposition} 	
%\end{document}
\begin{proof}
To formulate  Davies-Gaffney type estimates as proved in \cite{BHY17}, we consider the function 
\[
\xi: (0,\infty) \to (0, \infty), \quad \xi(r) = r \operatorname{arcsinh} r + 1 - \sqrt{1+r^2}.
\]
Then, \cite[Corollary~1.1]{BHY17} gives the following estimate for the heat kernel: there is $C > 0$, such that for each $x,y \in  X$ and $t > 0$, one has 
	\[
	p_t(x,y) \, \leq\, \frac{C}{\sqrt{m(x)m(y)}} \exp \left( -t \xi \Bigg( \frac{\varrho(x,y)}{t} \Bigg) \right).
	\]
	We consider this estimate for $y = o$  and integrate it with respect to $dt / t^{1-\alpha}$. We continue with the right hand side and compute by substitution with $t=\frac{|x|}{u}$ and $dt = -\frac{|x|}{u^2} du$ and setting $\Phi(u) = \xi(u) / u$  
	\begin{align*}
	I(x): = \int_0^{|x|}\exp \left( -t \xi \Bigg( \frac{|x|}{t} \Bigg) \right)\frac{dt}{t^{1-\alpha}}
	 = |x|^{\alpha} \int_{1}^{\infty} \frac{1}{u^{1+\alpha}} \exp\big( -|x|\Phi(u)\big) \, du.
	\end{align*}
	We observe further that the function $\Phi$ is monotonically increasing since 
	\begin{align*}
		\Phi^{\prime}(u) &= \frac{1}{\sqrt{u^2+1}} - \frac{1}{u^2} + \frac{u^{-2}}{\sqrt{u^2+1}} = \frac{u^2 + 1 - \sqrt{u^2+1}}{u^2\sqrt{u^2+1}} \geq 0.
	\end{align*}
	Consequently, we can estimate
	\begin{align*}
		I(x) &\leq |x|^{\alpha} \exp\big( -\xi(1)|x|\big) \int_{1}^{\infty} \frac{du}{u^{1+\alpha}}=  \frac{|x|^{\alpha}}{\alpha} \exp\big( -\xi(1)|x| \big).
	\end{align*}
%	Since $ \xi(r)\geq \frac12 r\log(r)$, there exists $C > 0$ such that
%	\[ 
%	I(x) \leq \frac{C}{\alpha}|x|^{\alpha} \exp\left(-\frac12 |x| \log |x| \right)
%	\]
%	 for all $x \in X$   and all $\alpha >0$. 
Our assumption on $m$ to be sub-exponentially bounded from below gives $\sqrt{m(o)m(x)} \geq K \exp(-\kappa|x|)$ for some constants $K >0$ and $ \kappa \in (0,\xi(1)) $ and  all $x\in X$. Combining these considerations with Davies-Gaffney type estimate above, we obtain 
	 \[
	 \int_0^{|x|} p_t(x,o) \frac{dt}{t^{1-\alpha}} \leq \frac{I(x)}{\sqrt{m(o)m(x)}}\leq \frac{C}{\alpha}|x|^{\alpha}e^{-c|x|}
	 \]
	 for some $0 < c < 1$. Since $\alpha$ is taken from a bounded interval, 
	 this gives the desired estimate, possibly after modifying the constants $c$ and $C$. 
\end{proof}

With the short time bounds at hand, we can now turn to the proof of Theorem~\ref{theorem:RieszAsymptotics}.

\begin{proof}[Proof of Theorem~\ref{theorem:RieszAsymptotics}] 
	The proof of the theorem  is based on splitting the integral defining $k_\alpha$ into a short time and a long time part. The short time part is controlled by the bounds provided by Proposition~\ref{prop:SU} and the long time part by the bounds in \eqref{e:GB}. We observe that the convergence of the integral at $0$ is always guaranteed since $\alpha > 0$ by assumption. Thus, finiteness of the Riesz kernel is determined and will be derived from the assumptions on the long time asymptotics.

	The analyticity statement will be argued along the way as follows: we observe that for any compact subset $ D $ of $(0,d/\beta)$ and  fixed $t>0$, $x,y\in X$ the function 
	$f:\alpha\mapsto p_t(x,y)t^{-1+\alpha}/\Gamma(\alpha) $ is analytic in a neighborhood of  $ D $. We will show that  $f$ admits an integrable majorant that is independent for $\alpha$  taken from a complex neighborhood of $D$, both in the short time and the long time region. This then gives analyticity of the integral by the Weierstra{\ss} theorem on uniformly convergent sequences of analytic functions.

%Pick $x \in X$ with $|x| \geq R_0$, where  $R_0 > 0$ is chosen according to~Proposition~\ref{prop:SU}. 
For $ x\in X $, we split the integral defining $k_\alpha$ into two regions
\begin{align*}
 k_\alpha(x)&=\frac{m(o)}{\Gamma(\alpha)}\int_0^\infty p_t(x,o) \, \frac{dt}{t^{1-\alpha}} = \frac{m(o)}{\Gamma(\alpha)} \left(\int_{0 }^{ |x|} \dots +\int_{|x| }^{\infty} \dots
\right),
\end{align*}
and denote the two integrals in the bracket by $I_1^\alpha(x)$ and $I_2^\alpha(x)$, respectively.

Invoking Proposition~\ref{prop:SU}, given an arbitrary compact interval $D \subseteq (0,d/\beta)$, there are $C, c > 	0$  such that for $\alpha \in D$, 
$$ I^{\alpha}_1(x) \leq \frac{C}{\alpha} \exp(-c|x|) ,$$ 
which  is superexponentially small for $|x|\to\infty$.  This also shows analyticity of the integral $I_1^{\alpha}(x)$ as well by the argument above.
 
For  the second region which consists of $t >  |x|$, we first consider the bounds in \eqref{e:GB} and set $ \mu=\beta/(\beta-1) $, $ \nu=1/(\beta-1) $ and $ \theta=d/\beta $.
We start by estimating the lower bound of $I_2^\alpha(x)$ which gives with the change of variables  $u = \frac{c_1|x|^\mu}{t^\nu}$, and consequently $t = \big(\frac{c_1|x|^{\mu}}{u}\big)^{1/\nu}$, $dt = -\frac{1}{\nu}\big({c_1|x|^{\mu}}\big)^{1/\nu}{u^{-1-1/\nu}}du$
\begin{align*}
    I_2^\alpha(x) &\ge C_1\int_{|x|}^\infty {e^{-c_1|x|^\mu/t^\nu}} \frac{dt}{t^{\theta+ 1-\alpha}}\\
&=\frac{C_1}{\nu} \left(c_1{|x|^\mu}\right)^{-\frac{\theta-\alpha}{\nu}}\int_{0}^{c_1|x|^{\mu-\nu}} e^{-u} u^{(\theta-\alpha)/\nu-1}du.
\end{align*}
We see that the integral on the right hand side converges if and only if $ \alpha<\theta=d/\beta $ which gives non-finiteness of $ k_\alpha(x) $ for $ \alpha \geq d/\beta $.
Using the definition of the Gamma function, \(
 \Gamma\left(z\right)=\int_0^\infty e^{-u} u^{z-1} du, 
\)
we obtain by resubstituting the choice of $\mu , \nu$ and $\theta$ that 
\begin{multline*}I_2^\alpha(x) \ge  \frac{C_1}{\nu}\left(c_1{|x|^\mu}\right)^{-\frac{\theta-\alpha}{\nu}} \left(   \Gamma\big((\theta-\alpha)/\nu\big) -\int_{c_1|x|^{\mu-\nu}}^{\infty} e^{-u} u^{(\theta-\alpha)/\nu-1}du\right) \\
	= {C_1}(\beta-1){\,{c_1}^{-(d/\beta-\alpha)(\beta-1)} \Gamma\left(\left(\frac{d}{\beta}-\alpha\right)(\beta-1)\right)} |x|^{-d+\alpha\beta}   -O\left(e^{-c^{\prime}|x|^c}\right)
\end{multline*}
for some constants $c,c^{\prime}>0$ that can be chosen independently from $\alpha$ in the given range.
Analogously, the upper bound of $I_2^\alpha(x)$ gives
\begin{align*}
	I_2^\alpha(x) &\le {C_2}(\beta-1){\,{c_2}^{-(d/\beta-\alpha)(\beta-1)} \Gamma\left(\left(\frac{d}{\beta}-\alpha\right)(\beta-1)\right)} |x|^{-d+\alpha\beta} .
\end{align*}
Hence, for $\alpha$ taken from a compact subset of $(0,d/\beta)$, we have the desired two sided bound independent of $\alpha$ and analyticity of the integral by the argument above. This also gives finiteness of $k_\alpha$ for $\alpha < d/\beta$.

Combining the estimates for the short time and the long time region, we obtain the desired asymptotics and analyticity of $k_\alpha$.
\end{proof}

We can now combine the asymptotics of the Riesz kernel with the Hardy weights from Proposition~\ref{p:exponents} to obtain asymptotics of the Hardy weights.

\begin{theorem}[Asymptotics of the fractional Hardy weights]\label{theorem:HardyAsymptotics}
Assume $m$ being sub-exponentially bounded from below  and $\sigma\in (0,1]$. 
 If  \eqref{e:GB} hold for $d\ge 1$, $ \beta\ge 2 $,  then the fractional Laplacian is transient if and only if $ \sigma<d/\beta $. In this case, there is a Hardy weight $ w_{\sigma,\al} $
 for all $ \alpha\in (\sigma,d/\beta) $  which satisfies the asymptotics
    \[  w_{\sigma,\al}(x) \asymp |x|^{-\beta\sigma},\]
    and the constants in the two-sided estimate can be chosen uniformly if $\alpha$ is taken from a compact subset of $(\sigma,d/\beta)$. 
\end{theorem}

\begin{proof}
	We notice that the Riesz kernel $k_\alpha$ is finite if and only if $\alpha \in (0,d/\beta)$ by Theorem~\ref{theorem:RieszAsymptotics}. 
	Hence, $ d/\beta = \alpha_{\max} $. Since $ k_\sigma=G^{\sigma} $ for $ \sigma\in (0,1) $,  we conclude that  transience is satisfied if and only if $ \sigma<d/\beta $.
	For $\sigma=1$, it is a direct consequence of~\eqref{e:GB} that $G(x,o)$ is finite if and only if $d/\beta > 1$. 
	Finally, the statement about the Hardy weight  follows by combining the above asymptotics of the Riesz kernel,  Theorem~\ref{theorem:RieszAsymptotics},   with Proposition~\ref{p:exponents}.
\end{proof}

Next, we turn to the question of criticality of the Hardy weights $w_{\sigma,\alpha}$ and show that the range of $\alpha$ decomposes into a positive critical, a null-critical and a subcritical regime. To this end, we need to assume some growth condition on the metric balls with respect to $\varrho$.

% We denote  balls of finite radius $r >0$ with respect to $\rho$ by \[B_r(x) = \big\{ y \in X \mid \varrho(y,x) \leq r  \big\}, \quad  x \in X. \]
%Given a graph $(b,c)$ over $(X,m)$ with a pseudo-metric $\rho$ and a fixed vertex $o\in X$, we will always denote $|x|=\varrho(x,o)$.

%We will say that the graph is \emph{Ahlfors regular of dimension} $d>0$  (with respect to $o$), i.e., there are constants $c_1,c_2>0, r_0$ such that for all $x\in X$ and $r\ge r_0$%\begin{align}\label{e:AR}\tag{AR}    c_1 r^{d} \leq m(B_r(o)) \leq c_2 r^{d}.\end{align}

We recall that given $d\ge 1 $, $ \beta\ge2$ and $ \sigma\in (0,1] $,  we   set
\(
    \alpha_0 = \frac{d +\sigma\beta}{2\beta}.
\)
If we assume 
$ \sigma \le d/\beta	 $, then we have $ \alpha_0\le d/\beta $ as well. Note that for $ \beta=2 $, we have that $ \alpha_0 $ is the middle point of $ \sigma $ and $ d/2 $.

\begin{theorem}[Asymptotics of the fractional Hardy weights]\label{theorem:HardyAsymptotics2}
%Let $(b,c)$ be a connected graph over $ (X,m) $ with exponentially lower bounded $ m $ and let $ \rho $ be a pseudo-metric on $X$. 
Assume $ m  $ is  sub-exponentially bounded from below and let  $d\ge 1$, $\beta\ge 2$ be  such that \eqref{e:AR}  and  \eqref{e:GB} hold.
Then, $ w_{\sigma,\al}={k_{\al-\sigma}}/{k_{\al}} $ is a strictly positive Hardy weight for all $ \sigma \in (0,1]$ with $ \sigma<d/\beta $ and $ \al\in (\sigma,d/\beta) $ and
	\begin{itemize}
		\item[(a)] for $\sigma <\alpha<\al_0 $, the weight $ w_{\sigma,\al} $ is positive critical,
		\item[(b)] for $ \al=\alpha_0 $, the weight $ w_{\sigma,\al} $ is null critical	%\end{itemize}
     %   If additionally  \eqref{e:GA} holds with the above parameters, then
      %  \begin{itemize}
	%	\item[(c)] for $ \alpha_0 <\al<d/2$, the weight $ w_{\sigma,\al} $ is subcritical.
	\end{itemize}
\end{theorem}

%\begin{remark} 	For the standard lattice graph defined over $(X,m) = (\Z^d,1)$ and parameters $\sigma \in (0,1]$ (and additionally $\sigma < d/2$ if $d \in \{1,2\}$), $\mu=2,\nu=1$ and $\alpha_0 = \frac{d}{4} + \frac{\sigma}{2}$, the range of $(\sigma,d/2)$ decomposes exactly according to these alternatives, see \cite[Theorem~3.5]{HKP}.\end{remark}

\begin{proof} The assertion~(a) follows directly from Theorem~\ref{t:poscritical}, Theorem~\ref{theorem:RieszAsymptotics} and annular Ahlfors regularity \eqref{e:AR} since
\begin{multline*}
    \sum_{X}k_{\alpha-\sigma }k_{\alpha}m \asymp 
	\sum_{x\in X}\frac{m(x)} {|x|^{ 2d -2\alpha\beta  +\sigma\beta  }}\asymp 
	\sum_{r=1}^\infty \frac{m(B_r)-m(B_{r-1})}{r^{ 2d -2\alpha\beta  +\sigma\beta  }}
	\\
	\asymp
	\sum_{r=1}^\infty m(B_r)\left( \frac{1}{r^{ 2d -2\alpha\beta  +\sigma\beta  }} -\frac{1}{(r+1)^{ 2d -2\alpha\beta  +\sigma\beta  }}\right)\\
	\asymp\sum_{r=1}^\infty  \frac{m(B_r)}{r^{ 2d -2\alpha\beta  +\sigma\beta +1 }}
	\asymp\sum_{r=1}^\infty
	\frac{r^{d}}{r^{ 2d -2\alpha\beta  +\sigma\beta +1 }},
	%\sum_{x\in X}\frac{m}{|x|^{2\frac{\mu}{\nu}(d/2-\alpha)+\frac{\mu}{\nu}\sigma}} \asymp \sum_{r=1}^\infty \frac{r^{d-1}}{r^{2\frac{\mu}{\nu}(d/2-\alpha)+\frac{\mu}{\nu}\sigma}}. 
\end{multline*}
where we used a telescoping sum argument and the mean value theorem.
The sum converges if and only if $ \alpha< (d+\sigma\beta)/2\beta=\alpha_0$.
    
We turn to (b) and show that for $ \al=\alpha_0 $, the weight $ w_{\sigma,\al} $ is  critical which gives null-criticality by (a). 

By Proposition~\ref{p:exponents}, we have  $k_{\alpha}=G^\sigma k_{\alpha-\sigma }$ and by Theorem~\ref{t:poscritical}  and (a) we have   $\sum_{X}k_{\alpha-\sigma }G^\sigma k_{\alpha-\sigma }m <\infty$  if and only if  $ \alpha<\alpha_0 $.  Thus, by Proposition~\ref{prop:equivalence_P_D_0}, we infer  $k_{\alpha}\in \mathcal{D}_0^\sigma$ for $\alpha<\alpha_0$. 
Thus,
\begin{align*}
	0&\leq (\mathcal{Q}^\sigma-w_{\sigma,\al_0})(k_{\alpha})=(\mathcal{Q}^\sigma-w_{\sigma,\al})(k_{\alpha})+(w_{\sigma,\al}-w_{\sigma,\al_0})(k_{\alpha})\\
	&=\sum_{X}(w_{\sigma,\al}-w_{\sigma,\al_0})k_{\alpha}^2 m.
	\end{align*}
Since $\alpha\mapsto \big(k_{\alpha-\sigma}/k_{\alpha}\big)(x)= w_{\sigma,\alpha}(x)$ is analytic for each fixed $x\in X$ by Theorem~\ref{theorem:RieszAsymptotics} (observe that $k_{\alpha}>0$ for $\alpha>0$), we apply Schwarz' theorem to obtain from the asymptotics of $w_{\sigma,\alpha}$ that
for $\alpha$ close to $\alpha_0$,
\begin{align*}
     |w_{\sigma,\al}(x)-w_{\sigma,\al_0}(x)| \le C|\alpha-\alpha_0| |x|^{- \sigma\beta}.
\end{align*}
Using Theorem~\ref{theorem:RieszAsymptotics} to estimate $k_\alpha$ and Ahlfors regularity \eqref{e:AR} gives by an analogous  argument as in (a) that for $\alpha$ close to $\alpha_0$, 
\begin{multline*}
 (\mathcal{Q}^\sigma-w_{\sigma,\al_0})(k_{\alpha})=(\mathcal{Q}^\sigma-w_{\sigma,\al})(k_{\alpha})	+  (w_{\sigma,\al}-w_{\sigma,\al_0} )(k_\alpha)\\
 \leq
  C|\alpha-\alpha_0| \sum_{x\in X}\frac{m(x)}{|x|^{ 2d - 2\alpha\beta + \sigma \beta }}
 \le
 % C|\alpha-\alpha_0|\sum_{r=1}^\infty   r^{-d-1 +2\alpha\beta - \sigma \beta} =
 C|\alpha-\alpha_0|\sum_{r=1}^\infty   r^{-1 + 2\beta(\alpha-\alpha_0)}
 \le C,
\end{multline*}
where the bound is uniform for $ \alpha $ close to $ \alpha_0 $.
Since  $k_{\alpha}$ is in $\mathcal{D}_0^\sigma$  and
$k_{\alpha}(x)\to k_{\alpha_0}(x)$ for all $ x\in X $  as $ \alpha\nearrow \alpha_0 $ by monotone convergence, 
we can extract a sequence $(e_n^\alpha)$ in $C_c(X)$ such that $ e_n^\alpha\to k_{\alpha} $, $\alpha<\alpha_0$ in the norm of $\mathcal{D}_0^\sigma$. Furthermore, by a Banach-Saks type argument, cf.~\cite[Theorem~5.3]{KLW}, \cite[Proposition~4]{KN} and a diagonal sequence argument, we obtain a null-sequence which shows that $ w_{\sigma,\al_0} $ is a critical Hardy weight by Theorem~\ref{lem:characterisaton_criticality}. By (a), the weight $ w_{\sigma,\al_0} $ is not positive critical, so it is null-critical. This shows (b).  
%Finally, to show (c) we assume \eqref{e:GA}. We observe that the constant $\Psi(\alpha)$ for $ \alpha_0 <\al<d/2$ is strictly decreasing in $\alpha$ cf \cite[Lemma~3.2]{FLS08} \Hmm{Correct reference?} and therefore, the Hardy weight $ w_{\sigma,\al} $ cannot be optimal at infinity as it is by a constant smaller than $ w_{\sigma,\al_0} $ outside of a finite set. Thus, for these $\alpha$ the Hardy weight $ w_{\sigma,\al} $ cannot be null-critical, see Proposition~\ref{lem:optimalnearinfty}. They can also not be positive critical by (a). Hence, they must be subcritical. This finishes the proof.
\end{proof}

\subsection{Examples and applications}\label{sec:examples}
In this section we discuss examples of graphs where our results for the fractional Laplacians apply. Note that the results of the previous sections require two-sided heat kernel bounds of Gaussian type and Ahlfors regularity. However, such estimates and bounds are only required with respect to one fixed vertex $o\in X$ and for large times. In consequence, our results yield optimal Hardy weights for the Laplacian and the fractional Laplacian on graphs.

However, the examples discussed here are  a small selection from the literature and we draw heavily from the excellent monongraph \cite{Barlowbook} which gives a comprehensive discussion of examples of graphs with Gaussian and sub-Gaussian heat kernel bounds. To simplify the presentation, we will restrict ourselves to graphs $ b $ over $ (X,\deg) $, where  $\deg(x) = \sum_{y \in X} b(x,y)$ for $x \in X$. The metric which we consider is the combinatorial metric, i.e., $\varrho(x,y)$ is the length of the shortest path connecting $x$ and $y$ from $ X$ denoted by $ d_c(x,y) $.

\subsubsection{Discrete vs.\@ continuous time and stability}
In the literature, two-sided Gaussian heat kernel bounds have been established for a variety of graphs. However, these bounds are mostly proven for discrete times in the form
\begin{align}\label{e:DGB}\tag{DHB}
C_1n^{-\frac{d}{\beta}}  e^{-c_1\left(\frac{|x|^\beta}{n}\right)^{\frac{1}{\beta-1}}} \leq  p^{{(n)}}(x,o) \leq C_2n^{-\frac{d}{\beta}}  e^{-c_2\left(\frac{|x|^\beta}{n}\right)^{\frac{1}{\beta-1}}} ,
\end{align}
where $p^{(n)}(x,y)$ is the $n$-step transition probability of the random walk associated to the graph Laplacian $L$ for a graph $b$ over $(X,\deg)$ and $ C_1,C_2,c_1,c_2>0 $ are constants. Indeed, one has to be a bit more careful in regards to the lower bound and bipartite graphs but we avoid these technicalities here and refer to \cite[Definition~4.14]{Barlowbook} for the technical details. Specifically, one can write $L=I-P$ and  $p^{(n)}(x,y)$ is given as the kernel of the operator $P^n$. 
%A natural choice of the parameters is
%\begin{align*}     \theta=\alpha/\beta, \quad \mu=\beta/(\beta-1), \quad \nu=1/(\beta-1),\end{align*}
%for some $\alpha \geq 1$ and $\beta \ge 2$. 
Gaussian bounds then correspond to the case $\beta=2$, while $\beta > 2$ corresponds to sub-Gaussian bounds which appear for fractal type graphs.

To derive continuous time bounds from discrete time bounds, one needs to assume some regularity on the graph. One says that the graph has \emph{controlled weights} if there is $C>0$ such that $$C b(x,y) \geq  \deg(x)$$ for all $x,y \in X$ with $b(x,y)>0$ and \emph{controlled degree} if there is $ C>0 $
$$C^{-1}\le \deg(x)\leq C $$
for all $x\in X$. In this case, one can pass from discrete time bounds to continuous time bounds for large times, see e.g. \cite{Barlowbook}.

\begin{proposition}[Theorems~5.24 and 5.25 in \cite{Barlowbook}] Let $b $ be a graph over $(X,\deg)$ with controlled weights and controlled degree. Let $d \geq 1$, $\beta \ge 2$ such that \eqref{e:AR} holds. Then, the following are equivalent:
\begin{itemize}
    \item[(i)] The graph satisfies discrete time heat kernel bounds \eqref{e:DGB}.
    \item[(ii)] The graph satisfies continuous time heat kernel bounds \eqref{e:GB}.
\end{itemize}    
\end{proposition}

So, the proposition above allows to derive the continuous time bounds \eqref{e:GB} from discrete time bounds \eqref{e:DGB} for large times. In consequence, our results apply to graphs satisfying discrete time heat kernel bounds as well.

\subsubsection{Graphs with Gaussian heat kernel bounds}
In this section, we discuss examples of graphs which are known to satisfy Gaussian heat kernel bounds \eqref{e:GB} for $d\ge 1$ and $\beta=2$.

The classical result in this regard is that Gaussian heat kernel bounds are equivalent to the combination of a volume doubling property and a Poincar\'e inequality, see e.g. \cite{Delmotte1999,Barlowbook}. We recall that a graph $b$ over $(X,\deg)$ satisfies the weak Poincar\'e inequality, cf.~\cite[Definition~3.28]{Barlowbook} if there are $C>0$ and $ \lambda\ge 1 $ such that for all $o\in X$, $r>0$ and $f:X\to \R$, one has 
\begin{align}\label{e:PI}\tag{P}
	\sum_{x\in B_r(o)}  \deg(x)(f(x)-f_{B_r(o)})^2 \leq C r^2 \sum_{x,y\in B_{\lambda r}(o)} b(x,y)(f(x)-f(y))^2,
\end{align}
where $f_{B} = \frac{1}{\deg(B)}\sum_{ B} \deg f$.% and $\deg(B) = \sum_{ B} \deg$. 

Another classical result is that under the condition of Ahlfor's regularity \eqref{e:AR}, Gaussian heat kernel bounds are equivalent to Nash inequalities, see e.g. \cite[Theorem~5.29]{Barlowbook}. We recall that a graph $b$ over $(X,\deg)$ satisfies a Nash inequality if there are $C>0$, $d\ge 1$ such that for all $f\in \ell^{1}(X,\deg)\cap\ell^{2}(X,\deg)$,
\begin{align}\label{e:N}\tag{N}
\|f\|_2^{2+4/d} \leq C \|f\|_1^{4/d} \mathcal{Q}(f).
\end{align}
Note that for $ d>1 $, the Nash inequality is equivalent to an isoperimetric inequality
\begin{align}\label{e:I}\tag{I}
	\inf_{W\subseteq X \mbox{\scriptsize{finite}}}\frac{b(\partial W)}{\deg(W)^{\frac{d-1}{d}}} >0,
\end{align}
where $ \partial W= (W\times X\setminus W)\cup (X\setminus W \times W) $, see \cite[Theorem~3.7]{Barlowbook}. Here, we denote $b(E) = \sum_{(x,y) \in E} b(x,y)$ for a set $E \subseteq X \times X$.

Note that while we require the heat kernel bounds \eqref{e:GB} to hold only with respect to one fixed vertex $o\in X$, the conditions in the proposition give the corresponding bounds for all $o\in X$.

\begin{proposition}[Theorem~6.18 and 6.19 in \cite{Barlowbook}]\label{prop:GBequ} Let $b $ be a graph over $(X,\deg)$ with controlled weights and controlled degree. Let $d \geq 1$.
\begin{itemize}
	\item [(a)] \eqref{e:GB} hold if and only if \eqref{e:AR} and \eqref{e:PI} hold (for all $ o\in X $).
	\item [(b)] If \eqref{e:AR} holds, then \eqref{e:GB} (for all $ o\in X $) hold if and only if \eqref{e:N} holds, which in the case $ d>1 $ is equivalent to \eqref{e:I}.
\end{itemize}		
\end{proposition}

Given the equivalences above, 
another important result in this regard is that certain heat kernel bounds are stable under rough isometries. A \emph{rough isometry} between two graphs $b_1$ over $(X_1,\deg_1)$ and $b_2$ over $(X_2,\deg_2)$ with controlled weights 
 is a map $\phi:X_1\to X_1$  such that there is a constant $C>0$ with
\begin{align*}
    C^{-1} d_1(x,y) - C \leq d_2(\phi(x),\phi(y)) \leq C d_1(x,y) + C,
\end{align*} 
for all $x,y \in X_1$,
\begin{align*}
    X_2 \subseteq \bigcup_{x \in X_1} B_C(\phi(x)),
\end{align*}
and 
\begin{align*}
    C^{-1} \deg_1(x) \leq \deg_2(\phi(x)) \leq C \deg_1(x),
\end{align*}
for all $x \in X_1$. Here,   $d_i$ is the graph distance on $X_i$, $i=1,2$, and $B_C(\phi(x))$ is the ball of radius $C$ around $\phi(x)$ with respect to the graph distance on $X_2$. 
A specific example of two rough isometric graphs $b_1$ over $(X,\deg_1)$ and $b_2$ over $(X,\deg_2)$ with controlled weights are so called {controlled weight perturbations}, i.e.,  if there is $C>0$ such that   
%\begin{align*}
$$      C^{-1} b_1\le b_2\leq C b_1.  $$
%\end{align*} 

\begin{proposition}[Theorem~6.19 in \cite{Barlowbook}]\label{prop:GBri} Let $b $ be a graph over $(X,\deg)$ with controlled weights and controlled degree. Let $d \geq 1$ be such that \eqref{e:AR} holds and let $ \beta=2 $. Then, 
    \eqref{e:GB} are stable under rough isometries % and controlled weight perturbations 
	 (with the same parameters $ d $ and $ \beta $).
\end{proposition}
%\begin{proof}     The stability of \eqref{e:DGB} under rough isometries is explicitly stated in  which gives the stability of \eqref{e:GB} by the proposition above. %The stability under controlled weight perturbations follows from  \cite[Theorem~6.31 ~(a)]{Barlowbook} combined with \cite[Theorem~6.31 ~(a)]{Barlowbook} \end{proof}

With these preparations above, we can now give explicit examples of graphs satisfying Gaussian heat kernel bounds and, therefore, our main results are applicable. One example are Cayley graphs of finitely generated groups of polynomial growth. By Gromov's theorem, such groups are virtually nilpotent \cite{Gromov,Kleiner}.
Secondly, we give the example of graphs satisfying a Bakry-Emery curvature-dimension type inequality $ CDE'(d_0,0) $ which is found in \cite[Definition~2.4]{HLLY} and says
\begin{align*}
\Gamma_2(f) - \Gamma_1\big(f,\frac{\Gamma_1(f)}{f}\big) \geq \frac{1}{d_0} f(x)^2 (\mathcal{L} \log f)^2,
\end{align*}
where $ \Gamma_1(f)= \frac{1}{2} \big( \mathcal{L} f^2 - 2f \mathcal{L} f \big) $ and $ \Gamma_2(f) = \frac{1}{2} \big( \mathcal{L} \Gamma_1(f) - 2\Gamma_1(f,\mathcal{L} f) \big) $.
A third example is that of semiplanar graphs with non-negative sectional curvature in the sense of \cite{HJL}. A semiplanar graph is embedded into a surface without self-intersection such that  each face is homeomorphic to a closed
disk with finite edges as the boundary. The sectional curvature is the angular defect at each vertex assuming the faces are regular polygons. We refer the reader to \cite{HJL} for the  details where it is also shown that such graphs with non-negative sectional curvature have at most quadratic volume growth.

\begin{corollary}
	Let $b$ be a graph over $(X,m)$ with controlled weights and controlled degree which  is roughly isometric to a graph  with controlled weights and controlled degree satisfying \eqref{e:AR} for $ d\ge 1 $ and either of the following conditions hold true:
	\begin{itemize}
	\item [(a)] The graph is a Cayley graph of a finitely generated group and $ d>1 $.
	\item [(b)] The graph has  non-negative Bakry-\'Emery curvature $CDE'(d_0,0)$ for some $ d_0>0 $ and $d \geq 1$.
	\item [(c)] The graph is semiplanar  with non-negative sectional curvature and $ d=2. $
	\end{itemize} 
	Then, for all $ \sigma\in (0,1] $, $ \sigma<d/2 $, the Hardy weight $$  w_{\sigma,\al_0}  \asymp |x|^{-2\sigma}$$ is an optimal Hardy weight for  $ \al_0=(d/2+\sigma)/2 .$
\end{corollary}
\begin{proof} One first observes that \eqref{e:AR} is stable under rough isometries, cf.~\cite[Exercise 4.16]{Barlowbook}.
	
	(a) By \cite[Theoreme 1]{CS}, see also \cite[Theorem~4.9]{Grigoryanbook}, we see that the graph satisfies \eqref{e:I}. Hence, by  Proposition~\ref{prop:GBequ} above, the graph satisfies \eqref{e:GB} with $d$ being the growth rate of the volume of balls with respect to the graph distance. Furthermore, the assumption on controlled degrees, gives that the measure is sub-exponentially bounded from below. Hence, the assumptions of Theorem~\ref{theorem:HardyAsymptotics2} are satisfied and the assertion follows.
	
	(b) In \cite[Theorem~2.2]{HLLY} it is shown that $ CDE'(d_0,0) $ implies discrete Gaussian heat kernel bounds \eqref{e:DGB}. Thus, the statement follows from Propositions~\ref{prop:GBequ},~\ref{prop:GBri} and  Theorem~\ref{theorem:HardyAsymptotics2}.

	(c) In \cite[Theorem 1.2]{HJL} the authors show that such graphs satisfy  a weak Poincar\'e inequality \eqref{e:PI}. Hence, by Proposition~\ref{prop:GBequ}, the graph satisfies \eqref{e:GB} with $d=2$. The assertion then follows from Theorem~\ref{theorem:HardyAsymptotics2}.
\end{proof}

We refer the reader who is interested in possibly unbounded graphs which have non-negative Bakry-Emery curvature of type $ CDE' $ to \cite{GLLY} where they show a parabolic Harnack inequality. This in turn can be seen to be equivalent to Gaussian heat kernel bounds even in the unbounded case \cite{BC}.

%For instance, such bounds are known to hold for graphs satisfying the so-called \emph{volume doubling property} and a \emph{Poincar\'e inequality}, see e.g. \cite{Delmotte1999,HK}. Examples of such graphs include \emph{Cayley graphs of finitely generated groups of polynomial growth} (in particular, lattices $\Z^d$), \emph{graphs with polynomial volume growth and non-negative Ollivier curvature}, \emph{graphs with polynomial volume growth and non-negative Bakry-\'Emery curvature}, and \emph{graphs satisfying certain isoperimetric inequalities}, see e.g. \cite{HLLY,MW,LPZ,CLP16}.

\subsubsection{Fractal-type graphs}
In this section, we discuss graphs which allow for heat kernel bounds beyond $ \beta
=2 $. Typically, these are fractal type graphs such as the Sierpinski gasket. Again we draw from \cite{Barlowbook} for an abstract characterization which uses the effective resistance given for a graph $ b $ over $ X $ as 
\begin{align*}
R(x,y)=\sup\{|f(x)-f(y)|^{2}\mid 
\mathcal{Q}(f)\leq1\}.
\end{align*}
We say that a graph $b$ over $(X,\deg)$ satisfies the \emph{resistance condition} if there are $C>0$, $\beta>0$ such that for all $x,y\in X$ and $r=d_c(x,y)>0$,
\begin{align}\label{e:RC}\tag{R}
	C^{-1}\frac{ r^{\beta} }{\deg(B_r(x))}\leq R(x,y) \leq C \frac{ r^{\beta} }{\deg(B_r(x))}.
\end{align}

Combining Proposition~\ref{prop:GBequ} with results from \cite{Barlowbook} gives the following characterization of heat kernel bounds. See also \cite{GT} for a characterization in terms of volume and Green function estimates.

\begin{proposition}[Theorem~6.9 and Corollary 6.10 in \cite{Barlowbook}]
	Let $b$ be a graph over $(X,\deg)$ with controlled weights and controlled degree. Let $\beta>d $. Then, the following are equivalent:
	\begin{itemize}
		\item[(i)] The graph satisfies heat kernel bounds \eqref{e:GB} with parameters $d$ and $\beta$.
		\item[(ii)] The graph satisfies Ahlfors regularity \eqref{e:AR} with parameter $ d $ and the resistance condition \eqref{e:RC} with parameter  $\beta$.
	\end{itemize}
	In this case, $ d\ge1 $ and  $2\le \beta\leq d+1 $.
\end{proposition}

In particular, the proposition above gives a characterization of heat kernel bounds for $\beta >2$ which are typical for fractal type graphs. 

\begin{figure}
	\includegraphics[width=\textwidth]{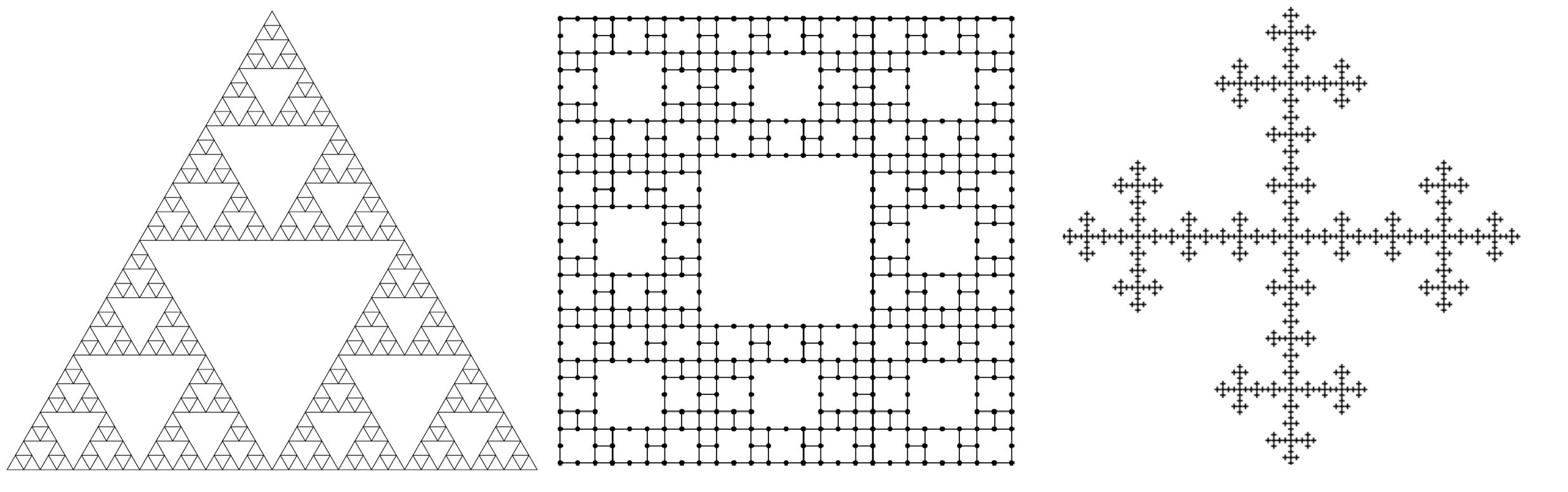}
	\caption{From left to right: Sierpinski gasket, Sierpinski carpet and Vicsek set.}
\end{figure}

\begin{example}
	(a) The \emph{Sierpinski gasket} graph satisfies \eqref{e:GB} with $d=\log 3/\log 2$ and $\beta=\log 5/\log 2$, see e.g. \cite[Corollary~6.11]{Barlowbook} or \cite{Jo}. Hence, for $ \sigma\in (0,\log 3/\log 5) $, the fractional Laplacian is transient and there is a Hardy weight $ w_{\sigma,\al} $ for all $ \alpha\in (\sigma,\log 3/\log 5) $  which satisfies the asymptotics
	\[  w_{\sigma,\al}(x) \asymp |x|^{- \sigma\log 5/\log 2}.\]
which is optimal for $$  \alpha=\alpha_0=\frac{\log 3+ \sigma\log 5 }{2\log 5 }    $$
by Theorem~\ref{theorem:HardyAsymptotics2}.

(b) In \cite{BB}, the authors study the \emph{Sierpinski carpet} graph which is a planar graph with holes and show that it satisfies \eqref{e:GB} with for certain values of $ d $ and $ \beta $, where in particular $ \beta >2 $ is only shown to exist but not accessible to computation.

(c) The \emph{Vicsek set} is for example studied in \cite{Zhou}.
We assume a branching of $ 2 $, which means there are 4 arms in the Vicsek set. Then the parameters are given as $ d=\log 5/\log 3 $ with $ \beta=\log 15/\log 3 $ which gives transience of the fractional Laplacian for $ \sigma\in (0,\log 5/\log 15) $ and there exists a Hardy weight $ w_{\sigma,\al} $ for all $ \alpha\in (\sigma,\log 5/\log 15) $  which satisfies the asymptotics
	\[  w_{\sigma,\al}(x) \asymp |x|^{-\sigma \log 15 /\log 3}\]
	which is optimal for $$  \alpha=\alpha_0=\frac{\log 5+ \sigma \log 15}{2\log 15 }    $$
by Theorem~\ref{theorem:HardyAsymptotics2}.
\end{example}

\textbf{Acknowledgement.} The authors acknowledge the financial support of the DFG. The second and the third author are grateful for the hospitality and support  of  the IIAS.

\bibliographystyle{alpha}
\bibliography{mynewbib}

%\printbibliography
\end{document}